\documentclass[letter,notitlepage,twocolumn,superscriptaddress,longbibliography,aps,nofootinbib]{revtex4-1}

\usepackage[utf8]{inputenc}
\usepackage{braket}
\usepackage{amsmath,amsthm,amssymb,amsfonts}
\usepackage{graphicx}
\usepackage{hyperref}
\usepackage{comment}
\usepackage{color}
\usepackage{times}

\usepackage{soul} 

\usepackage{amsmath,amsthm,amssymb}
\usepackage[dvipsnames]{xcolor}
\usepackage{colortbl}
\usepackage{hyperref}
\usepackage{braket}
\usepackage{bbm}
\usepackage{bm}
\usepackage{longtable}
\usepackage[bb=boondox]{mathalfa}
\usepackage{adjustbox}
\usepackage{array}
\usepackage{multirow}
\usepackage{tabularx}
\usepackage{multirow}

\usepackage{float}

\usepackage{tikz,tikz-3dplot}
\usetikzlibrary{shapes,arrows,patterns,cd}

\usepackage[caption=false]{subfig}

\newtheorem{proposition}{Proposition}
\newtheorem{corollary}{Corollary}

\newtheorem{example}{Example}

\newcommand{\ccaption}[2]{\caption[#1]{\textit{#1.} #2}}

\definecolor{dred}{rgb}{.4,0,0}
\definecolor{lred}{rgb}{1,.8,.8}

\newcommand{\tr}{\operatorname{tr}} 
\newcommand{\Tr}{\operatorname{Tr}} 
\renewcommand{\Im}{\operatorname{Im}} %
\renewcommand{\Re}{\operatorname{Re}} %

\usepackage{array}
\newcolumntype{C}[1]{>{\centering\arraybackslash}p{#1}}

\usetikzlibrary{decorations.markings,arrows,decorations.pathreplacing,calc,intersections,through,backgrounds}

\definecolor{Gray}{gray}{0.95}
\newcolumntype{a}{>{\columncolor{Gray}}c}

\newcommand{\ii}{\mathrm{i}}
\newcommand{\ie}{{\it i.e.},\ }
\newcommand{\eg}{{\it e.g.},\ }

\newcommand{\id}{\mathbb{1}} 

\newcommand{\comm}[2]{\left[{#1},{#2}\right]}

\newcommand{\nn}{\nonumber\\}

\newcommand{\ff}{\mathrm{f}}
\newcommand{\bb}{\mathrm{b}}

\newcommand{\ee}[1]{\mathrm{e}^{#1}} 
\newcommand{\CC}{\mathbb{C}}

\begin{document}

\title{Entanglement dualities in supersymmetry}

\author{Robert H. Jonsson}
\email{robert.jonsson@mpq.mpg.de}
\affiliation{Max-Planck-Institut für Quantenoptik, Hans-Kopfermann-Str.~1, 85748 Garching, Germany}

\author{Lucas Hackl}
\email{lucas.hackl@unimelb.edu.au}
\affiliation{School of Mathematics and Statistics \& School of Physics, The University of Melbourne, Parkville, VIC 3010, Australia}
\affiliation{QMATH, Department of Mathematical Sciences, University of Copenhagen, Universitetsparken 5, 2100 Copenhagen, Denmark}

\author{Krishanu Roychowdhury}
\email{krishanu.1987@gmail.com}
\affiliation{Department of Physics, Stockholm University, SE-106 91 Stockholm, Sweden}

\begin{abstract}
We derive a general relation between the bosonic and fermionic entanglement in the ground states of supersymmetric quadratic Hamiltonians. For this, we construct canonical identifications between bosonic and fermionic subsystems. Our derivation relies on a unified framework to describe both bosonic and fermionic Gaussian states in terms of so-called linear complex structures $J$. The resulting dualities apply to the full entanglement spectrum between the bosonic and the fermionic systems, such that the von Neumann entropy and arbitrary Renyi entropies can be related. 
We illustrate our findings in one and two-dimensional systems, including the paradigmatic Kitaev honeycomb model. 
While typically supersymmetry preserves features like area law scaling of the entanglement entropies on either side, we find a peculiar phenomenon, namely, an amplified  scaling of the entanglement entropy (``super area law") in bosonic subsystems when the dual fermionic subsystems develop almost maximally entangled modes.
\end{abstract}

\maketitle


\section{Introduction}
As a long-established concept in quantum physics, supersymmetry (SUSY) finds applications in a wide range of fields from particle physics to condensed matter in both relativistic and nonrelativistic settings \cite{gol1989extension, ramond1971dual, neveu1971factorizable, sourlas1985introduction, baer2006weak}. In a nutshell, SUSY posits a fundamental equivalence between the two classes of elementary particles with distinct statistics. Mathematically, it maps the fermionic degrees of freedom to the bosonic ones and vice versa. From this perspective, they are equivalent and dubbed {\it superpartners} of each other. 

While normally SUSY is conceived as a symmetry in quantum field theories, it as well applies to much simpler models of quantum mechanics such as harmonic oscillators or the hydrogen atom \cite{cooper1983aspects, cooper1995supersymmetry, kirchberg2003algebraic, gangopadhyaya2017supersymmetric}. 
The SUSY Hamiltonian $\hat{H}$ can be constructed from a generating operator $\hat{\cal Q}$ (also called the supercharge operator) which, for the harmonic oscillator problem, takes a remarkably simple form $\hat{\cal Q}=\sqrt{\omega}\hat{b}^\dagger \hat{c}$ 
where $\hat b$ ($\hat c$) denotes the bosonic (fermionic) annihilation operator. The corresponding SUSY Hamiltonian
\begin{align}
 \hat{H} = \{\hat{\cal Q}, \hat{\cal Q}^\dagger\} = \omega(\hat{b}^\dagger \hat{b} + \hat{c}^\dagger \hat{c}) \equiv \hat{H}_{\bb} + \hat{H}_{\ff}
 \label{eq:susyosc}
\end{align}
then decomposes into two simple quadratic Hamiltonians: one for a bosonic oscillator ($\hat{H}_{\bb}$) and the other for a fermionic one ($\hat{H}_{\ff}$). When it comes to dealing with real bosons or fermions, a Hermitian form of the generating operator  $\hat{\cal Q}=\hat{\cal Q}^\dagger$
(and accordingly, $\hat{H}=\hat{\mathcal Q}^2$) is useful, as  also is the case for the present work.

Such a simple setting is readily amenable to accommodate multiple bosonic and fermionic modes, or in other words, systems of free (noninteracting) bosons and fermions (in  the  continuum or on a lattice) if the generating operator $\hat{\cal Q}$ involves the bosonic and fermionic operators to linear order \cite{lawler2016supersymmetry, attig2019topological}, as shown in the previous harmonic oscillator example and also will be demonstrated later. The resulting partner Hamiltonians (referred to as $\hat{H}_{\bb}$ and $\hat{H}_{\ff}$ for bosons and fermions, respectively) are isospectral in their one-particle excitations except for zero modes. Inclusion of zero modes in SUSY has, in addition, a topological aspect (referred to as ``Witten index'' \cite{witten1982constraints} and interpreted in several other contexts, e.g., see \cite{kane2014topological}) and has been studied to a great extent; however, that discussion is not relevant to this work.

Ground states of a quadratic Hamiltonian (bosonic or fermionic) garner special attention as they provide a fertile ground to trace several properties of the system, which they are part of, analytically. These states are also known as {\it Gaussian states}~\cite{wang2007quantum,weedbrook2012gaussian,adesso2014continuous,shi2018variational,Derezinski:2013dra}. The study of the von Neumann bipartite entanglement entropy plays a central role in the quantum foundations of statistical mechanics
\cite{deutsch_91,srednicki_94,rigol_dunjko_08,d2016quantum, gogolin2016equilibration,deutsch2018eigenstate,goldstein_lebowitz_06, popescu_short_06,tasaki_98,polkovnikov2011colloquium,vidmar2017entanglement,liu2018quantum,vidmar2018volume,hackl2019average,Vidmar:2017pak,bianchi2019typical,lydzba2020eigenstate}, in quantum information theory \cite{bennett1998quantum,eisert2006entanglement,Hayden:2006,Hayden:2007cs,Sekino:2008he,Hosur:2015ylk,Roberts:2016hpo,Fujita:2017pju,Lu:2017tbo,Fujita:2018wtr} and condensed matter dedicated to classifying novel states of matter, particularly those with topological quantum order \cite{kitaev2006topological,levin2006detecting, furukawa2007topological,yao2010entanglement, isakov2011topological,depenbrock2012nature, jiang2012identifying,gong2014emergent, roychowdhury2015z}. While measuring entanglement is numerically costly for a generic quantum state, it greatly simplifies for the Gaussian states \cite{sorkin1983entropy,peschel2003calculation}.

The main result of this work is a duality between the eigenvalues of reduced density operators in the bosonic and the fermionic system, \ie the so-called entanglement spectra. For Gaussian states, these spectra are fully encoded in the eigenvalues $\pm\ii\lambda$ of the so-called restricted complex structure $J$, where $\lambda_\bb\in[1,\infty)$ for bosons and $\lambda_\ff\in[0,1]$ for fermions. In supersymmetric systems, the charge operator $\hat{\mathcal{Q}}$ provides an identification between the bosonic and the fermionic system, so that picking a subsystem on the bosonic side automatically defines a related subsystem on the fermionic side and vice versa. Our key finding is that, under identification, we have $\lambda_\bb=1/\lambda_\ff$ and $J_\bb=-J^{-1}_\ff$, where we use $\bb$ and $\ff$ to refer to the bosonic and fermionic structure, respectively.

Applying our results to examples, we also discuss consequences of the derived duality for the entanglement entropy in Gaussian states related by SUSY.
Though not always, entanglement entropy often turns out to be a sufficient measure (among  others) of the entanglement information encoded in a quantum state \cite{horodecki2009quantum, laflorencie2016quantum, bengtsson2017geometry}.
In fact, in a number of strongly correlated systems, this quantity serves as a smoking gun to identify topological quantum order in the ground states. 
Examples include Kitaev's celebrated model of Majorana fermions on a honeycomb lattice \cite{kitaev2006anyons}.  
In earlier works \cite{attig2019topological}, the bosonic SUSY analog of this model has been realized and shown to inherit the topological properties from its fermionic partner. We will also regard this model here, as one of our examples to illustrate the aspects of entanglement dualities considering the SUSY-related Gaussian states.

Generally speaking, for noncritical ground states in $d$ dimensions (for both fermionic and bosonic systems), the entanglement entropy of a subsystem $A$ obeys the so-called ``area law'' (for a review, see \cite{eisert2010colloquium, laflorencie2016quantum} and references therein)
\begin{align}
 {\cal S}(A) \propto L^{d-1} + \dots,
 \label{eq:ent_area_law1}
\end{align}
meaning that, in the thermodynamic limit, the leading order contribution to the entanglement entropy of $A$ with the rest of the system scales with its surface area $L^{d-1}$ when $L$ denotes the linear dimension of $A$. 
For critical states, however, the ellipses in \eqref{eq:ent_area_law1} can contain sublinear corrections (\eg logarithmic corrections for free fermions), and for topologically ordered states, a universal constant called ``topological entanglement entropy.''

The identification provided by the supercharge $\hat{\mathcal Q}$ facilitates a natural connection between a subsystem in one lattice and a subsystem in the superpartner lattice. {\it A priori}, this identification does not warrant a local subsystem in one system to get mapped to a localized subsystem in its superpartner system. However, we will show that, even when well-localized subregions are identified of both lattices, the scaling of the entanglement entropy of the dual supersymmetric subsystem can be very different -- on the bosonic side, it can drastically exceed the area law exhibited by the original fermionic subsystem.

In summary, this study extends the concept of SUSY beyond a spectral mapping between (supersymmetric) quadratic Hamiltonians to discuss the general identification of fermionic and bosonic supersymmetric Gaussian systems, their subsystems, and entanglement spectra as implied by the supercharge operator. Exemplifying lattice models in one dimensions (1D) and two dimensions (2D), we investigate the locality properties of these identification maps and their consequences in the context of entanglement area laws.
In doing so, we employ the idea of {\it K{\"a}hler structure}, which brings the bosonic and fermionic Gaussian states within a unified frame to work in. A further  merit of this approach lies in treating the involved geometric structures independent of their matrix representation in a given basis, as discussed at length, \eg in  \cite{hackl2019minimal,hackl2020bosonic,windt2020local}. 
\\

The article is structured as follows:  In Sec.\ \ref{sec:review-Gaussian-states}, we review the unified K{\"a}hler structure formalism to describe bosonic and fermionic Gaussian states and apply it to supersymmetric quadratic Hamiltonians, where a charge operator induces an identification map at the classical phase space level. 
In Sec.\ \ref{sec:entanglement-duality}, we explore how the entanglement entropies in the bosonic and fermionic systems are related and introduce a general theorem on their entanglement spectra. In Sec.\ \ref{sec:discussion}, we summarize our key findings complemented by lattice models as applications and discuss future work.

\section{Gaussian states and supersymmetry}\label{sec:review-Gaussian-states}
In this section, we review the unified formalism that treats both bosonic and fermionic Gaussian states on the same footing. For this, we present a hands-on introduction to the formalism of\ \cite{hackl2020bosonic}, which can be consulted for a more rigorous exposition. Other reviews of Gaussian states include\ \cite{weedbrook2012gaussian}.

\subsection{Bosonic and fermionic Gaussian states}\label{eq:Gaussian-states}
We consider a bosonic or fermionic system with $N$ degrees of freedom described by a Hilbert space $\mathcal{H}$. We can always find a basis of creation and annihilation operators which we denote as $\hat b_i$ and $\hat b_i^\dagger$ for bosons, and as $\hat c_i$ and $\hat c_i^\dagger$ for fermions, but we use $\hat a_i$ and $\hat a_i^\dagger$ in expressions valid for both bosons and fermions (see Table~\ref{tab:notations}). These operators satisfy the canonical commutation or anti-commutation relations 
\begin{equation}
    \begin{aligned}
    [][\hat{b}_i,\hat{b}_j^\dagger]&=\delta_{ij}\,, && \textbf{(bosons)}\\
    \{\hat{c}_i,\hat{c}_j^\dagger\}&=\delta_{ij}\,. && \textbf{(fermions)}
    \end{aligned}\label{eq:CCR-CAR}
\end{equation}
Out of these, we can construct a set of $2N$ Hermitian operators
\begin{equation}
    \begin{aligned}
    \hat{q}_i=\tfrac{1}{\sqrt{2}}(\hat{b}_i^\dagger + \hat{b}_i)\,\,\quad\hat{p}_i&=\tfrac{\ii}{\sqrt{2}}(\hat{b}_i^\dagger-\hat{b}_i)\,, && \textbf{(bosons)}\\
    \hat{\gamma}_i=\tfrac{1}{\sqrt{2}}(\hat{c}_i^\dagger + \hat{c}_i)\,,\quad\hat{\eta}_i&=\tfrac{\ii}{\sqrt{2}}(\hat{c}_i^\dagger-\hat{c}_i)\,, && \textbf{(fermions)}
    \end{aligned}
\end{equation}
which satisfy the commutation or anti-commutation relations
\begin{equation}
    \begin{aligned}
    [][\hat{q}_i,\hat{q}_j]&=[\hat{p}_i,\hat{p}_j]=0\,,\,\,\,\,\,\,\,\,\,\,[\hat{q}_i,\hat{p}_j]=\ii \delta_{ij}, \hspace{-1mm}&& \textbf{(bosons)}\\
    \{\hat{\gamma}_i,\hat{\gamma}_j\}&=\{\hat{\eta}_i,\hat{\eta}_j\}=\delta_{ij}\,,\, \{\hat{\gamma}_i,\hat{\eta}_j\}=0. && \textbf{(fermions)}
    \end{aligned}
\end{equation}
For bosons, these operators are commonly called quadrature operators (generalized positions and momenta), while for fermions, they are called the Majorana operators.

Up to normalization, there is a unique state $\ket{0}\in\mathcal{H}$, such that $\hat{a}_i\ket{0}=0~\forall i$, which is called the vacuum state with respect to our choice of operators. An orthonormal basis of $\mathcal{H}$ can then be constructed by successively applying creation operators on $\ket0$,
\begin{align}
    \ket{n_1,\dots,n_N}=\prod^N_{i=1}\frac{(\hat{a}_i^\dagger)^{n_i}}{\sqrt{n_i!}}\ket{0}\,,\label{eq:Fock-basis}
\end{align}
where $n_i\in\mathbb{N}$ for bosons and $n_i=0,1$ for fermions.

We can now collect the $2N$ operators to form the vector
\begin{align}\label{eq:quantization_map_quads_Majorana}
    \hat{\xi}^a\equiv\begin{cases}
    (\hat{q}_1,\dots,\hat{q}_N,\hat{p}_1,\dots,\hat{p}_N) &\textbf{(bosons)}\\
    (\hat{\gamma}_1,\dots,\hat{\gamma}_N,\hat{\eta}_1,\dots,\hat{\eta}_N) &\textbf{(fermions)}
    \end{cases}\,,
\end{align}
where we have the index $a=1,\dots,2N$ (later, we use Latin indices exclusively for bosons and Greek indices for fermions, but for now we use Latin indices for both). It is well known that, for both bosons and fermions, any operator $\mathcal{O}$ can be described as a power series in $\hat{\xi}^a$ or as a limit of such a series. For many physically relevant operators, this series will be finite and of low order. The canonical commutation or anti-commutation relations in terms of $\hat\xi^a$ read
\begin{equation}\label{eq:CCR_CAR_in_terms_of_forms}
\begin{aligned}
    [][\hat{\xi}^a,\hat{\xi}^b]&=\ii \Omega^{ab}\,, && \textbf{(bosons)}\\
    \{\hat{\xi}^a,\hat{\xi}^b\}&=G^{ab}\,, && \textbf{(fermions)}
\end{aligned}
\end{equation}
where $\Omega^{ab}$ is called the symplectic form and $G^{ab}$ is a metric. With respect to our choice of basis in \eqref{eq:quantization_map_quads_Majorana}, they are represented by the matrices
\begin{align}
    \Omega\equiv\begin{pmatrix}
    0 & \id\\
    -\id & 0
    \end{pmatrix}\quad\text{and}\quad G\equiv\begin{pmatrix}
    \id & 0\\
    0 & \id
    \end{pmatrix}\,,\label{eq:OmG-standard}
\end{align}
and will play an important role in later formulas.

\begin{table}[t]
    \centering
    \renewcommand{\arraystretch}{1.75}
    \begin{tabular}{ p{1.7cm} p{3.2cm} p{3.4cm}} 
        \toprule
         & {\bf Real basis} & {\bf Complex basis}   \\
        \colrule
        {\bf Bosons} & Quadrature operators \newline
        $\hat{\xi}_\bb\equiv(\hat{q}_j,\hat{p}_k)$
        & 
        CCR operators  \newline
        $(\hat{b}_j,\hat{b}^\dagger_k)$
        \\
        {\bf Fermions} & Majorana operators\newline
        $\hat{\xi}_\ff\equiv (\hat \gamma_j,\hat \eta_k) $
            & CAR operators 
            \newline
        $(\hat{c}_j,\hat{c}^\dagger_k)$
        \\
        {\bf Unified} & Hermitian operators $\hat{\xi}$ &  Ladder operators $(\hat{a}_j,\hat{a}^\dagger_k)$ \\
        \botrule
    \end{tabular}
    \ccaption{Overview of notations for operator bases}{Listed are real (self-adjoint) and complex operator bases for bosons and fermions, as well as a unified notation used throughout this work. For an $N$-mode quantum system, indices are in the range $j,k \,{\in}\, \{1,\dots,N\}$. The creation and annihilation operators, in a complex basis, satisfy \emph{canonical commutation/anti-commutation relations} (CCR/CAR).}
    \label{tab:notations}
\end{table}

We define\footnote{
Here, we restrict to Gaussian states with $\braket{J|\hat{\xi}^a|J}=0$, \ie the 1-point correlation function vanishes. However, the formalism extends to also include displacements $z^a=\braket{J|\hat{\xi}^a|J}$ for bosons, as explained in\ \cite{hackl2020bosonic}.
} a Gaussian state $\ket{J}\in\mathcal{H}$ as the solution of\footnote{
Note that\ \eqref{eq:def-Gaussian} only fixes $\ket{J}$ up to a complex phase. This does not cause any problems when considering individual Gaussian states, where the complex phase is unphysical. However, if considering superpositions of Gaussian states $\ket{J}+\ket{\tilde{J}}$, we would need to parametrize explicitly how the respective complex phases are related.}
\begin{align}
    \frac{1}{2}(\delta^a{}_b+\ii J^a{}_b)\hat{\xi}^b\ket{J}=0\,.\label{eq:def-Gaussian}
\end{align}
As shown in\ \cite{hackl2020bosonic}, a solution of~\eqref{eq:def-Gaussian} exists only if $J^2=-\id$ and the following compatibility conditions are satisfied:
\begin{itemize}
    \item For bosons, $G^{ab}:=-J^a{}_c\Omega^{cb}$ is a metric, \ie symmetric and positive definite.
    \item For fermions, $\Omega^{ab}:=J^a{}_cG^{cb}$ is a symplectic form, \ie anti-symmetric and non-degenerate. 
\end{itemize}
The matrix $J$ is called a linear complex structure.

\emph{In\ \eqref{eq:def-Gaussian} and the rest of this paper, we use Einstein's summation convention,\footnote{
All our equations with indices are fully basis independent and compatible with Penrose's abstract index notation\ \cite{penrose1984spinors}. In fact, we can even use complex bases, such as $\hat{\xi}^a\equiv(\hat{a}_1,\dots,\hat{a}_N,\hat{a}_1^\dagger,\dots,\hat{a}_N^\dagger)$ (see~\cite{hackl2020bosonic}).} where a sum is implied over repeated indices (index contraction). The position of the index indicates if it can be contracted with vectors $v^a\in V$ in phase space or dual vectors $w_a\in V^*$ in dual phase space. Objects with two indices are often written as matrices, where matrix multiplication is the same as contraction over adjected indices. This may require a transpose, \eg $\Omega^{ac}J^b{}_c$ needs to be written as $(\Omega J^\intercal)^{ab}=\Omega^{ac}(J^\intercal)_c{}^b$ to make the indices $c$ adjacent.}

The above relations introduce for every Gaussian state $\ket{J}$ the object $G^{ab}$ for bosons and $\Omega^{ab}$ for fermions, such that we have in both cases a so-called \emph{Kähler structure}: This is a triplet $(G,\Omega,J)$ such that
\begin{align}
    G^{ab}=-J^a{}_c\Omega^{cb}\quad\Leftrightarrow\quad \Omega^{ab}=J^a{}_cG^{cb}\,,\label{eq:GOmJ-relations}
\end{align}
the equivalence following from $J^2=-\id$. Moreover, we have $J\Omega J^\intercal=\Omega$ and $JGJ^\intercal=G$.

This definition of Gaussian states, unifying bosons and fermions, may appear surprising to readers more familiar with the definition of Gaussian states in terms of covariance matrices or Bogoliubov transformations. However, as shown in\ \cite{hackl2020bosonic}, these definitions are fully equivalent, as we review in the following.

\subsubsection{Covariance matrix}
The covariance matrix of a quantum state $\ket{\psi}$ with $\braket{\psi|\hat{\xi}^a|\psi}=0$ is defined as\footnote{Some authors use a different normalization or sign. The extension to states with $\braket{\psi|\hat{\xi}^a|\psi}\neq 0$ is also straight-forward and explained in\ \cite{hackl2020bosonic}.}
\begin{align}
    \Gamma^{ab}=\begin{cases}
    \braket{\psi|\hat{\xi}^a\hat{\xi}^b+\hat{\xi}^b\hat{\xi}^a|\psi} & \textbf{(bosons)}\\
    -\ii\braket{\psi|\hat{\xi}^a\hat{\xi}^b-\hat{\xi}^b\hat{\xi}^a|\psi} & \textbf{(fermions)}
    \end{cases}\,,
\end{align}
\ie the covariance matrix is exactly the expression that is not already fixed by the canonical commututation or anti-commutation relations. Given a Gaussian state $\ket{J}$ with associated Kähler structures $(G,\Omega,J)$, it follows from\ \eqref{eq:def-Gaussian} that we have the 2-point function
\begin{align}
    C_2^{ab}:=\braket{J|\hat{\xi}^a\hat{\xi}^b|J}=\frac{1}{2}(G^{ab}+\ii\Omega^{ab})\,,
\end{align}
To prove this, we define $\hat{\xi}^a_\pm=\frac{1}{2}(\delta^a{}_b\mp\ii J^a{}_b)\hat{\xi}^b$, which depend on $J$. With this, we find $\hat{\xi}^a=\hat{\xi}^a_++\hat{\xi}^a_-$, and we have $\hat{\xi}^a_-\ket{J}=0$ and $\bra{J}\hat{\xi}^a_+=0$, due to\ \eqref{eq:def-Gaussian}. This implies
\begin{align}
    C_2^{ab}=\braket{J|\hat{\xi}^a_-\hat{\xi}^b_+|J}=\begin{cases}
    \braket{J|[\hat{\xi}^a_-,\hat{\xi}^b_+]|J} & \textbf{(bosons)}\\
    \braket{J|\{\hat{\xi}^a_-,\hat{\xi}^b_+\}|J} & \textbf{(fermions)}
    \end{cases}
\end{align}
due to $\braket{J|\hat{\xi}^b_+\hat{\xi}^a_-|J}=0$. Finally, the commutator or anti-commutator above can be evaluated using\ \eqref{eq:CCR_CAR_in_terms_of_forms} to be
\begin{equation}
\begin{aligned}
\hspace{-2mm}[\hat{\xi}^a_-,\hat{\xi}^b_+]&=\tfrac{1}{4}(\id+\ii J)^a{}_c\ii\Omega^{cd}(\id-\ii J)^b{}_d\,,&&\textbf{(bosons)}\\
\hspace{-2mm}\{\hat{\xi}^a_-,\hat{\xi}^b_+\}&=\tfrac{1}{4}(\id+\ii J)^a{}_c G^{cd}(\id-\ii J)^b{}_d\,,&&\textbf{(fermions)}
\end{aligned}
\end{equation}
which in both cases combines to $\frac{1}{2}(G+\ii \Omega)$ via\ \eqref{eq:GOmJ-relations}.

We can reverse this argument to use $C_2^{ab}$ (and thus the covariance matrix $\Gamma^{ab}$ contained in it) of a general state $\ket{\psi}$, with $\braket{\psi|\hat{\xi}^a|\psi}=0$, to check if $\ket\psi$ is a Gaussian state \ie $\ket\psi=\ket{J}$ and find $J$. For this, we first compute $G^{ab}=2\Re\braket{\psi|\hat{\xi}^a\hat{\xi}^b|\psi}$ and $\Omega^{ab}=2\Im\braket{\psi|\hat{\xi}^a\hat{\xi}^b|\psi}$ and then invert\ \eqref{eq:GOmJ-relations} to compute
\begin{align}
J^a{}_b=\Omega^{ac}(G^{-1})_{cb}\,.\label{eq:J-from-G-Om}
\end{align}
One can then show\ \cite{hackl2020bosonic} that $J^2=-\id$ is necessary and sufficient for $\ket{\psi}$ to be the Gaussian state $\ket{J}$, \ie a solution of \eqref{eq:def-Gaussian}. However, if $J^2\neq -\id$, $\ket{\psi}$ is not a Gaussian state.

\subsubsection{Bogoliubov transformations}
These transformations map Gaussian states into Gaussian states, hence, are also termed \emph{Gaussian transformations}. For a Gaussian state $\ket{J}$ annihilated by a set of annihilation operators $\hat{a}'_i$, \ie $\hat{a}'_i\ket{J}=0$, the following transformation relates them to the original $\hat{a}_i$ 
\begin{align}
    \hat{a}'_i=\sum_j(\alpha_{ij}\hat{a}_i+\beta_{ij}\hat{a}_j^\dagger)\,,
\end{align}
where the matrix elements $\alpha_{ij}$ and $\beta_{ij}$ characterize the transformation. Defining a Gaussian state $\ket{J_0}$ as the state annihilated by all $\hat{a}_i$, \ie $\hat{a}_i\ket{J_0}=0$, we can use\ \eqref{eq:J-from-G-Om} and\ \eqref{eq:OmG-standard} to compute that $J_0$ is represented by the matrix
\begin{align}
    J_0\equiv\begin{pmatrix}
    0 & \id\\
    -\id & 0
    \end{pmatrix}\,,\label{eq:J0-standard}
\end{align}
from which we deduce the resulting Bogoliubov transformed state $\ket{J}$ with $J=MJ_0M^{-1}$, where the matrix $M$ is \cite{windt2020local}
\begin{align}
    M=\begin{pmatrix}
    \Re\alpha+\Re\beta & \Im\beta-\Im\alpha\\
    \Im\alpha+\Im\beta & \Re\alpha-\Re\beta
    \end{pmatrix}\,.
\end{align}
The matrix $M$ is a group element of the symplectic group $\mathrm{Sp}(2N,\mathbb{R})$ for bosons or the orthogonal group $\mathrm{O}(2N,\mathbb{R})$ for fermions, which induces the unitary representation of Gaussian transformations on the Hilbert space\ \cite{hackl2020bosonic}.

\begin{example}
The simplest bosonic Gaussian state is the ground state of the harmonic oscillator with Hamiltonian $\hat{H}=\frac{1}{2}(\hat{p}^2+\omega^2\hat{q}^2)$ that takes the form
\begin{align}
    \ket{J}=\frac{1}{\cosh{\tfrac{\rho}{2}}}\sum^\infty_{n=0}(-\tanh{\tfrac{\rho}{2}})^n\ket{2n}
\end{align}
with respect to the basis\ \eqref{eq:Fock-basis} and $\rho=\log{\omega}$. Its covariance matrix $\Gamma^{ab}=G^{ab}$ and complex structure as
\begin{align}
    \Gamma=G\equiv\begin{pmatrix}
    \omega & 0\\
    0 & \frac{1}{\omega}
    \end{pmatrix}\quad\text{and}\quad J\equiv\begin{pmatrix}
    0 & \omega\\
    -\frac{1}{\omega} & 0
    \end{pmatrix}
\end{align}
with respect to the basis $\hat{\xi}^a\equiv(\hat{q},\hat{p})$.\\
The simplest fermionic Gaussian states are the basis states $\ket{J_+}=\ket{0}$ and $\ket{J_-}=\ket{1}$, which are also the only Gaussian states for a single degree of freedom. Their covariance matrices $\Gamma_\pm=\Omega_\pm$ and complex structures $J_\pm$ happen to coincide in the basis $\hat{\xi}^a\equiv(\hat{q},\hat{p})$ as
\begin{align}
    \Gamma_\pm=\Omega_\pm\equiv J_\pm\equiv\begin{pmatrix}
    0 & \pm 1\\
    \mp 1 & 0
    \end{pmatrix}\,.
\end{align}
\end{example}

In summary, this section has reviewed how bosonic and fermionic Gaussian states can be efficiently described in a unified formalism using the triplet $(G,\Omega,J)$ of Kähler structures. Physical properties, such as expectation values or entanglement entropies can be directly computed from them.

\subsection{Supercharge operator and supersymmetric Gaussian states}
We will now consider a system that contains both, bosonic and fermionic degrees of freedom. We denote the bosonic operators by $\hat{\xi}^{a}_{\mathrm{b}}$ and the fermionic ones by $\hat{\xi}^{\alpha}_{\mathrm{f}}$, where we use Latin letters for bosons and Greek letters for fermions. The commutation and anti-commutation relations then read
\begin{align}
    [\hat{\xi}^{a}_{\mathrm{b}},\hat{\xi}^{b}_{\mathrm{b}}]=\ii \Omega^{ab}_\bb\quad\text{and}\quad\{\hat{\xi}^{\alpha}_{\mathrm{f}},\hat{\xi}^{\beta}_{\mathrm{f}}\}=G^{\alpha\beta}_\ff\,,
\end{align}
while the bosonic and the fermionic operators commute $[{\hat\xi^a_\bb},{\hat\xi^\alpha_\ff}]=0$.

The SUSY transformation between the bosonic and the fermionic degrees of freedom can be generated by a Hermitian supercharge operator \cite{lawler2016supersymmetry}
\begin{align}
    \hat{ \mathcal{Q}}=R_{\alpha a}\hat{\xi}^{\alpha}_{\mathrm{f}}\hat{\xi}^a_{\mathrm{b}}\,,
\end{align}
with a real-valued $R$. As mentioned already in~\eqref{eq:susyosc}, this supercharge defines a supersymmetric Hamiltonian
\begin{align}
\begin{split}
    \hspace{-3mm}\hat{H}&=\tfrac{1}{2}\{\hat{\mathcal{Q}},\hat{\mathcal{Q}}\}
    =\tfrac{1}{2}h^\bb_{ab}\hat{\xi}^{a}_{\mathrm{b}}\hat{\xi}^{b}_{\mathrm{b}}+\tfrac{\ii}{2}h^\ff_{\alpha\beta}\hat{\xi}^{\alpha}_{\mathrm{f}}\hat{\xi}^{\beta}_{\mathrm{f}}\, \equiv \hat{H}_\bb+\hat{H}_\ff
    ,\label{eq:sym-hamiltonian}
\end{split}
\end{align}
which splits into a bosonic part $\hat H_\bb$ and a fermionic part $\hat H_\ff$. Their Hamiltonian forms are
\begin{align}
    h^\ff_{\alpha\beta}&= R_{\alpha a}\Omega^{ab}R^\intercal_{b\beta}
        \,,\label{eq:sym-hamiltonian3}\\
    h^\bb_{ab}&= R^\intercal_{a\alpha}G^{\alpha\beta}R_{\beta b}
    \,,\label{eq:sym-hamiltonian4}
\end{align}
which satisfy  $h^\ff_{\alpha\beta}=-h^\ff_{\beta\alpha}$ and $h_{ab}^\bb=h_{ba}^\bb$. 
Note that the full Hamiltonian's ground state energy $E_0=\ii R_{\alpha a}R_{\beta b} G^{\beta\alpha}\Omega^{ba}=\ii\tr (GR\Omega^\intercal R^\intercal)=0$ vanishes, as the bosonic and the fermionic contributions cancel each other. 

The excitation spectrum of $\hat{H}_\bb$ and $\hat{H}_\ff$ can be derived by diagonalizing  the Lie generators $K_{\mathrm{b}}$ and $K_{\mathrm{f}}$, defined via the relations\footnote{Alternatively, one can also exploit the Heisenberg equation of motion leading to $\tfrac{d}{dt}\hat{\xi}_{\mathrm{b}}^a=\ii[\hat{H},\hat{\xi}^a_{\mathrm{b}}]=\ii(K_{\mathrm{b}})^a{}_b\hat{\xi}_{\mathrm{b}}^b$ and similarly for $\hat{\xi}_{\mathrm{f}}^\alpha$.}
\begin{align}
    [\hat{H},\hat{\xi}_{\mathrm{b}}^a]=(K_{\mathrm{b}})^a{}_b\hat{\xi}_{\mathrm{b}}^b\quad\text{and}\quad[\hat{H},\hat{\xi}_{\mathrm{f}}^\alpha]=(K_{\mathrm{f}})^\alpha{}_\beta\hat{\xi}_{\mathrm{f}}^\beta\,.
\end{align}
One can show\ \cite{hackl2020bosonic} that these matrices are Lie algebra elements satisfying
\begin{align}
    K_{\mathrm{b}}\Omega=-\Omega K_{\mathrm{b}}^\intercal\quad\text{and}\quad K_{\mathrm{f}}G=-G K_{\mathrm{f}}^\intercal\,,
\end{align}
which implies $K_{\mathrm{b}}\in\mathfrak{sp}(2N,\mathbb{R})$ and $K_{\mathrm{f}}\in\mathfrak{so}(2N,\mathbb{R})$. Using the relations\ \eqref{eq:CCR-CAR} allows us to compute them explicitly as 
\begin{align}
    (K_{\mathrm{b}})^a{}_b&=\tfrac{1}{2}\Omega^{ac}(h_{cb}+h_{bc})=\Omega^{ac} R^\intercal_{c\alpha}G^{\alpha \beta}R_{\beta b}\,,\\
    (K_{\mathrm{f}})^\alpha{}_\beta&=\tfrac12 G^{\alpha\gamma}\left(h_{\gamma\beta}-h_{\beta\gamma}\right)
    =G^{\alpha\gamma} R_{\gamma a}\Omega^{ab}R^\intercal_{b\beta}\,.
\end{align}
From this, it is evident that $K_\bb$ and $K_\ff$ are isospectral except for the degeneracy of potential zero eigenvalues. 

The ground state of $\hat{H}$ is given by the tensor product
\begin{align}\label{eq:groundstate_ket}
    \ket{\mathrm{GS}}=\ket{J_{\mathrm{b}}}\otimes\ket{J_{\mathrm{f}}}\,,
\end{align}
where the associated $J_{\mathrm{b}}$ and $J_{\mathrm{f}}$ are computed from the generators as\ \cite{hackl2020bosonic,bianchi2018linear}\footnote{ Note that applying a function $f(K)$, such as the absolute value, to a diagonalizable matrix $K=U^{-1}DU$, where $D$ is a diagonal matrix containing the eigenvalues of $K$, is equivalent to applying $f$ to its eigenvalues, \ie $f(K)=U^{-1}f(D)U$. }
\begin{align}
    J_\bb=\left|K_\bb^{-1}\right|K_\bb\quad\text{and}\quad
    J_\ff=\left|K_\ff^{-1}\right|K_\ff\,.\label{eq:J-from-K}
\end{align}
These formulas may be surprising at first sight, but they can be readily checked using a basis, where the individual normal modes of $\hat{H}$ decouple. In this basis, we have
\begin{align}
    \hat{H}=\sum_i\frac{\omega_i}{2}(\hat{n}_i^{\bb}+\hat{n}_i^{\ff})\,,
\end{align}
where $\hat{n}_i=\hat a_i^\dagger\hat a_i$  are the normal mode number operators and $\omega_i$ are the one-particle excitation energies. Note that, due to $h^\bb_{ab}$ being positive, all $\omega_i$ are positive, and we choose $\hat{n}_i^\ff$, such that excitations increase energy. If we go into the associated basis $\hat{\xi}^a$, where $\hat{n}_i^{\bb}=\frac{1}{2}\left((\hat{q}^\bb_i)^2+(\hat{p}^\bb_i)^2\right)$ and $\hat{n}_i^{\ff}=\ii \hat{\gamma}_i\hat{\eta}_i$, the matrix representations of the generators are
\begin{align}
K_\bb\equiv K_\ff\equiv\oplus_i\begin{pmatrix}
0 & \omega_i\\
-\omega_i & 0
\end{pmatrix}\,.\label{eq:spectrumkb}
\end{align}
In this specific basis, $J_\bb$ and $J_\ff$ assume the standard form from\ \eqref{eq:J0-standard}, which then implies\ \eqref{eq:J-from-K}.

\begin{example}\label{ex:SUSY-Hamiltonian}
The simplest supersymmetric Hamiltonian consists of one bosonic and one fermionic degree of freedom. The respective supercharge operator is given by
\begin{align}
    \hat{\mathcal{Q}}=\hat{q}\hat{\gamma}+\hat{p}\hat{\eta}\,,
\end{align}
for which we find the Hamiltonian
\begin{align}
    \hat{H}=\hat{\mathcal{Q}}^2=\tfrac{1}{2}(\hat{q}^2+\hat{p}^2)+\tfrac{\ii}{2}(\hat{\gamma}\hat{\eta}-\hat{\eta}\hat{\gamma})\,.
\end{align}
(equivalent forms of $\hat{\cal Q}$ and $\hat H$ in terms of complex bosonic and fermionic operators are shown in the introduction). The associated Lie algebra generators are then given by
\begin{align}
    K_\bb\equiv K_\ff\equiv\begin{pmatrix}
    0 & 1\\
    -1 & 0
    \end{pmatrix}\,,
\end{align}
and the associated ground state is $\ket{\mathrm{GS}}=\ket{0_\bb}\otimes\ket{0_\ff}$.
\end{example}

\subsection{Supersymmetric identification maps}\label{sec:susy-identification}
We introduced supersymmetric Hamiltonians through the supercharge operator $\hat{\mathcal{Q}}$ as $\hat{H}=\hat{H}_\bb+\hat{H}_\ff$, where $\hat{H}_\bb$ and $\hat{H}_\ff$ have identical one-particle spectra. Both the bosonic and the fermionic part are described classically by phase spaces $V_\bb\simeq\mathbb{R}^{2N}$ and $V_\ff\simeq\mathbb{R}^{2N}$ (with the corresponding dual spaces denoted by $V^\ast_\bb$ and $V^\ast_\ff$), respectively, such that $V_\bb$ is equipped with the symplectic form $\Omega_\bb$, and $V_\ff$ is equipped with a metric $G_\ff$. 
The respective other structure in each space, \ie a metric $G$ on $V_\bb$ and a symplectic form $\Omega$ on $V_\ff$, is defined by the ground state $\ket{J}=\ket{J_\bb}\otimes\ket{J_\ff}$ of $\hat{H}$.

In this section, we now use the supercharge $\hat{\mathcal Q}$ to construct linear maps  $L_1:V_\bb\to V_\ff$ and $L_2:V_\ff\to V_\bb$ between the two phase spaces, that identify the spaces in such a way that the symplectic forms and metrics are mapped onto each other.

Under the above assumption that $R_{a\alpha}$ is real, the supercharge operator $\hat{\mathcal{Q}}=R_{\alpha a}\hat\xi^\alpha_\ff\hat\xi^a_\bb$ induces the \emph{supersymmetric identification maps} $T_1:V_\bb\to V_\ff$ and $T_2:V_\ff\to V_\bb$ as
\begin{align}\label{eq:identification_maps}
    (T_1)^\alpha{}_a=G^{\alpha\beta}R_{\beta a}\quad\text{and}\quad (T_2)^a{}_\alpha=\Omega^{ab}R^\intercal_{b\alpha}\,.
\end{align}
These are related to the Lie generators noting $K_\bb= T_2 T_1$ and $K_\ff=T_1 T_2$. 
Hence, $T_2$ maps the eigenvectors of $K_\ff$
(in $V_{\ff,\CC}$, the complexification on $V_\ff$) to the eigenvectors of $K_\bb$ with the same eigenvalue, and for $T_1$, the analogous holds:
\begin{equation}
\begin{aligned}
    K_\bb v_\bb &=\pm\ii \lambda_v v_\bb &&\Rightarrow &&K_\ff T_1v_\bb= \pm\ii\lambda_v  T_1v_\bb\\
    K_\ff w_\ff& =\pm\ii \lambda_w w_\ff&& \Rightarrow && K_\bb T_2w_\ff= \pm\ii\lambda_w  T_2w_\ff
\end{aligned}
\end{equation}
If only the spaces $V_\bb$ and $V_\ff$ are given, each equipped with Kähler structures $(G,\Omega,J)$, then there exists a large class of potential identification maps;\footnote{Given an identification map $T_1: V_\bb\to V_\ff$, we can define a new identification $T_1'=U_\ff T_1 U_\bb$, where both $U_\bb: V_\bb\to V_\bb$ and $U_\ff: V_\ff\to V_\ff$ need to preserve the respective Kähler structures. This implies that $U_\bb$ and $U_\ff$ form a representation of the group $\mathrm{U}(N)$. In our case, we also would like that $T_1$ maps $K_\bb$ onto $K_\ff$, which implies that the respective symmetry group will depend on the degeneracy of the one-particle spectrum. If $K_\bb$ (and thus also $K_\ff$) has $m$ distinct eigenvalue pairs $\pm\ii\lambda_i$ with degeneracy $d_i$ such that $\sum_{i=1}^m d_i=N$, the resulting symmetry group will be $\mathrm{U}(d_1)\times \dots \times \mathrm{U}(d_m)$. Only if the Hamiltonian is fully degenerate with $N$ eigenvalue pairs $\pm \lambda$, this will lead to the maximal symmetry group $\mathrm{U}(N)$ of possible identification maps $T_1'$.} however, the choice of $R_{\alpha b}$ fixes this freedom.

We can use the supersymmetric identification maps to construct normalized identification maps $L_1: V_\bb\to V_\ff$ and $L_2: V_\ff\to V_\bb$ as
\begin{align}\label{eq:J1_J2_definition}
    L_1=|K^{-1}_\ff|^{1/2} T_1,\quad L_2=|K^{-1}_\bb|^{1/2} T_2,
\end{align}
(where the form of $L_1$ was identified in\ \cite{attig2019topological}). 
These have the property that their products exactly reproduce the linear complex structures \begin{align}
    L_1L_2=J_\ff\quad\text{and}\quad L_2L_1=J_\bb,
\label{eq:J1J2_JfJb}
\end{align}
of the ground state of $\hat H$. 

To see this, it is convenient to work in the eigenbases of the generators $K_\bb$ and $K_\ff$. Let $v^{(\pm k)}\in V_{\bb,\CC}$ denote a basis of eigenvectors  of $K_\bb$ with eigenvalues $\pm \ii \lambda_k$. Then $\{T_1 v^{(\pm k)}\}$ is a basis of $V_{\ff,\CC}$ diagonalizing $K_\ff$. In fact, with respect to these bases $K_\bb$ and $K_\ff$ are represented by the same matrix. Accordingly, also $|K^{-1}_\ff|^{1/2}$ and $|K^{-1}_\bb|^{1/2}$ are represented by the same matrices. From this follows, in particular,
\begin{align}
    L_1&=|K^{-1}_\ff|^{1/2} T_1=T_1 |K^{-1}_\bb|^{1/2},\\
    L_2&=|K^{-1}_\bb|^{1/2} T_2=T_2 |K^{-1}_\ff|^{1/2}\,,
\end{align}
and hence, we have
\begin{align}
\begin{split}
    L_1L_2&=|K^{-1}_\ff|^{1/2} T_1 |K^{-1}_\bb|^{1/2} T_2\\
    &= |K^{-1}_\ff| T_1 T_2= |K^{-1}_\ff|K_\ff= J_\ff\,,
\end{split}\\
\begin{split}
    L_2L_1&=|K^{-1}_\bb|^{1/2} T_2|K^{-1}_\ff|^{1/2} T_1\\
    &=|K^{-1}_\bb|  T_2  T_1=|K^{-1}_\bb|K_\bb =J_\bb\,.
\end{split}
\end{align}

In the following, we use the identification maps to associate both linear observables and quadratic forms between the two supersymmetric partner systems. 
For this, it is important to note that, since the identification maps and their inverses  act on the phase spaces, \ie  they act on upper indices from the left, their corresponding transposes act on the dual phase spaces, \ie  on lower indices, as
\begin{align*}
    &V_\bb\xrightarrow{L_1}  V_\ff, && V_\ff \xrightarrow{(L_1)^{-1}}  V_\bb, && V^*_\bb\xrightarrow{(L_1^\intercal)^{-1}}  V^*_\ff, && V^*_\ff\xrightarrow{ L_1^\intercal } V^*_\bb,
\\
   & V_\ff \xrightarrow{ L_2}  V_\bb, && V_\bb\xrightarrow{(L_2)^{-1}} V_\ff, && V^*_\ff\xrightarrow{ (L_2^\intercal)^{-1} }  V^*_\bb, && V^*_\bb\xrightarrow{ L_2^\intercal }  V^*_\ff\, .
\end{align*}
For example, let $\hat s=s_a \hat \xi^a_\bb$ be a linear operator on the bosonic system. Then 
\begin{align}
L_2(\hat s)=s_a (L_2)^a{}_\alpha \hat\xi^\alpha_\ff,
\end{align}
is the linear fermionic operator associated to it by the identification map $L_1$. Analogously, if $\hat r=r_\alpha\hat \xi^\alpha_\ff$ is a fermionic operator, the identification map $L_1$ associates the bosonic operator
\begin{align}
L_1(\hat r)=r_\alpha (L_1)^\alpha{}_a\hat \xi^a_\bb
\end{align}
with it. In this sense, the identification maps always identify corresponding pairs of eigenmodes of the SUSY Hamiltonian with each other: If we diagonalize the SUSY Hamiltonian as
\begin{align}
\hat{\mathcal{Q}}^2=\sum_i \omega_i \left(\hat b_i^\dagger \hat b_i+\hat c_i^\dagger \hat c_i\right)
\end{align}
then, assuming that all $\omega_i$ are different, we always have
\begin{align}
L_1\left(\hat c_{i}\right) =\ee{\ii\phi_{i,1}}\hat b_i,\quad L_2\left(\hat b_i\right)=\ee{\ii\phi_{i,2}}\hat c_{i}
\end{align}
for all $i=1,...,N$, because of~\eqref{eq:J1_J2_definition}. 
Also, due to~\eqref{eq:J1J2_JfJb} the  complex phases are such that $\ee{\ii\phi_{i,1}} \ee{\ii\phi_{i,2}}=-\ii$, since $J_\bb(\hat b_i)=-\ii\hat b_i$ and $J_\ff(\hat c_i)=-\ii\hat c_i$ as follows from~\eqref{eq:GOmJ-relations} and~\eqref{eq:OmG-standard} (expressed in the complex bases).

\begin{example}
The supercharge operator $\hat{\mathcal{Q}}$ introduced in Example\ \ref{ex:SUSY-Hamiltonian} induces the rather simple identification maps represented by the matrices
\begin{align}
L_1\equiv \begin{pmatrix} 1&0\\0&1\end{pmatrix},\quad L_2\equiv \begin{pmatrix}0&1\\-1&0\end{pmatrix}.
\end{align}
Accordingly the Hermitian mode operators are identified as
\begin{equation}
\begin{aligned}
L_1(\hat \gamma)&=\hat q,& L_1(\hat \eta)&=\hat p\,, \\
L_2(\hat q)&=\hat \eta,& L_2(\hat p)&=-\hat \gamma\,.
\end{aligned}
\end{equation}
\end{example}

\subsection{Application: supersymmetric Kitaev chain}\label{sec:Kitaev_1D_example}

\begin{figure}[tb]
\begin{center}
  \includegraphics[width=.95\linewidth]{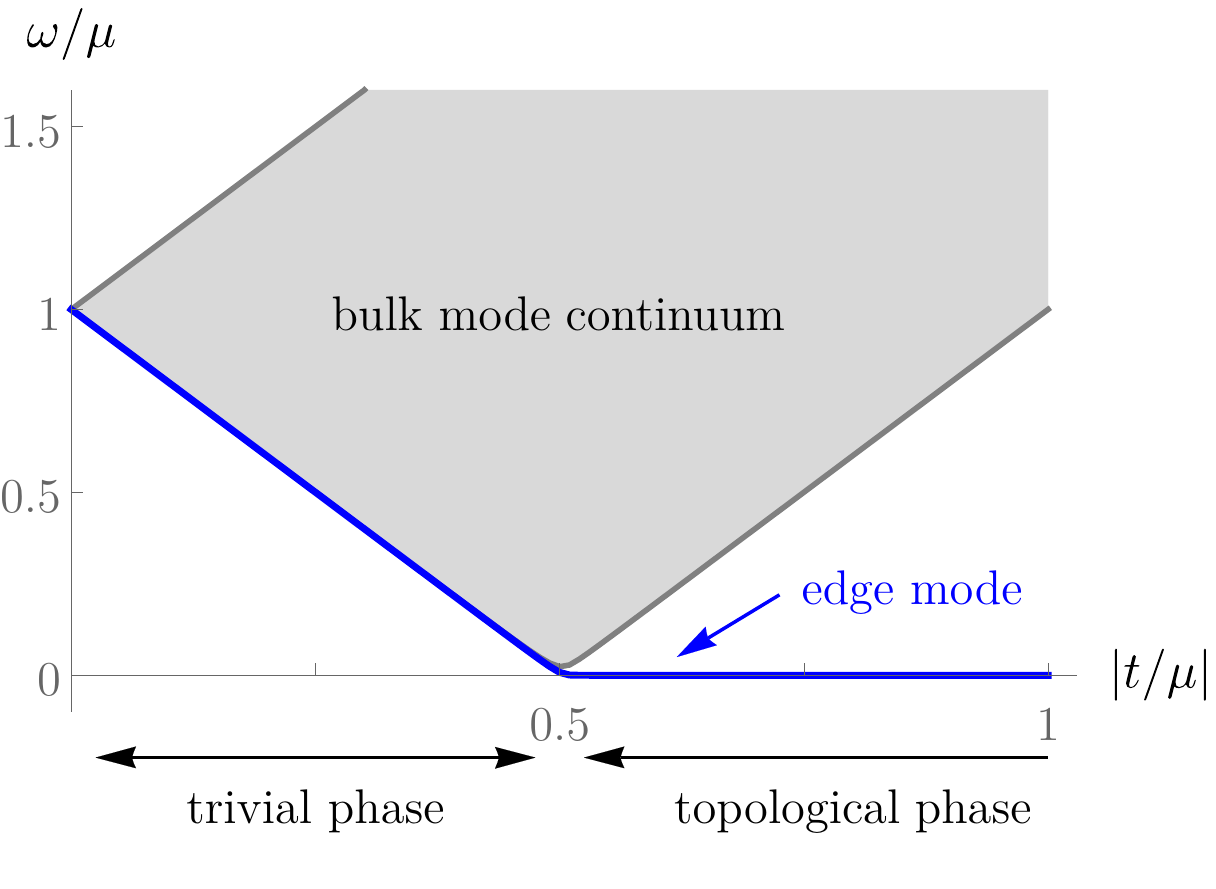}
\end{center}\vspace{-4mm}
\caption{Spectrum of the Kitaev chain with open ends. The system is in a trivial phase for $|t/\mu|<1/2$, critical at $t/\mu=\pm1/2$, and topological otherwise, with edge modes appearing. 
} 
\label{fig:kitaev_1d_spectrum} 
\end{figure}

In this section, we choose the well-known Kitaev chain\ \cite{kitaev2001unpaired} of $N$ sites with open boundary conditions as a concrete application for the formalism above and  investigate the physical properties of the identification maps. In our construction, the supersymmetric partner of the Kitaev chain resembles the Kane-Lubensky (KL) chain \cite{kane2014topological}. We  are interested in addressing the question: To what extent do the identification maps preserve the localization properties of operators, when mapping them from one system to its SUSY partner?

The form of the fermionic Kitaev chain Hamiltonian which we study is obtained by considering a real pairing, and setting its magnitude equal to the hopping ($t$) in the original model proposed in\ \cite{kitaev2001unpaired}:
\begin{align}\label{eq:Kitaev_SUSY_Hf}
    \hat H_\ff &=  \frac{\mu}2 \sum_{i=1}^N  \big(\hat c_i^\dagger \hat c_i- \hat c_i \hat c_i^\dagger\big)   
    +  t \sum_{i=1}^{N-1}  \big(\hat c_{i+1}^\dagger \hat c_i  - \hat c_{i+1}^\dagger \hat c_i^\dagger  +\mathrm{H.c.}\big),
\end{align}
where $\mu$ denotes the chemical potential. A supercharge which generates this Hamiltonian as the fermionic part of $\hat{\mathcal{Q}}^2=\hat H_\ff+\hat H_\bb$ is given by
\begin{align}\label{eq:Kitaev_1D_supercharge}
    \hat{\mathcal{Q}} &=  \sqrt\mu\sum_{i=1}^N   \hat c_i\hat b_i^\dagger  
    +\frac{t}{\sqrt\mu} \sum_{i=1}^{N-1} \left( \hat c_{i}\hat b_{i+1} + \hat c_{i}\hat b_{i+1}^\dagger \right) +\mathrm{H.c.}
    \nn&=  \sqrt{\mu}\sum_{i=1}^N\left( \hat\gamma_i \hat q_i + \hat\eta_i\hat p_i\right)+ \frac{2t}{\sqrt\mu} \sum_{i=1}^{N-1} \hat\gamma_{i+1}\hat q_i\,
\end{align}
Its  bosonic part resembles the KL chain, a well-studied model in topological mechanics \cite{kane2014topological}:
\begin{align}\label{eq:Kitaev_SUSY_Hb}
    &\hat H_\bb = \frac\mu2\sum_{i=1}^N \hat p_i^2  +\frac{4t^2+\mu^2}{2\mu} \sum_{i=2}^{N} \hat q_i^2 +\frac\mu2\hat q_1^2 +2t \sum_{i=1}^{N-1}\hat q_i\hat q_{i+1}
    \nn&=  \frac\mu2\left(\hat b_1\hat b_1^\dagger+\hat b_1^\dagger \hat b_1\right)
        + \sum_{i=2}^{N}\left[ \tfrac\mu2 \left(1+\tfrac{2t^2}{\mu^2}\right) \left(\hat b_i^\dagger\hat b_i+\hat b_i\hat b_i^\dagger\right) \right.
    \nn& \left.\qquad\qquad + t \left(\hat b_{i-1}\hat b_{i}^\dagger+\hat b_{i-1}\hat b_{i}\right) +\frac{t^2}\mu \left(\hat b_i\hat b_i\right) +\mathrm{H.c.}\right]\,.
\end{align}
Denoting the  energy eigenmodes of the system with primed operators, the SUSY Hamiltonian can be diagonalized as
\begin{align}\label{eq:1d_Kitaev_H_diagonal}
    \hat{\mathcal{Q}}^2=\hat H_\ff+\hat H_\bb = \sum_{i=1}^N \omega_i \left(\hat{b}'{}^\dagger_i\hat{b}'_i+\hat {c}'_i{}^\dagger \hat{c}'_i \right).
\end{align}

Figure~\ref{fig:kitaev_1d_spectrum} schematically shows the spectrum of the Kitaev chain, which is in a trivial phase for $|{t}/{\mu}|<1/2$ and in a topological phase otherwise. 
The bulk gap closes at the critical point ${t}/{\mu}=\pm1/2$ in the limit of large $N$. The trivial phase is featureless; all eigenmodes together form a bulk mode continuum. However, as the system enters the topological phase for $|{t}/{\mu}|>1/2$, an edge mode gradually separates from the continuum and stabilizes at zero energy (albeit with an exponentially small gap with $N$) as a telltale signature of the topological phase.
On the fermionic side, \ie for the Kitaev chain, the edge modes are localized at both ends of the chain.
In contrast, on the bosonic side, \ie for the KL chain, they are localized only at one end (here the left end) of the chain. 
For completeness, we mention that, in the KL chain, there exists a nonlinear zero mode (soliton) that can reverse the location of the edge mode \cite{chen2014nonlinear}; however, that falls beyond the ambit of the present setting. 
The localization of the edge mode at the boundaries of the chain is exponential, in the sense that, when writing the edge mode operator as $\hat c'_{N}=\sum_j \alpha_j\hat c_j+\beta_j\hat c_j^\dagger$ or $\hat b'_{N}=\sum_j \alpha_j\hat b_j+\beta_j\hat b_j^\dagger$, the quantities $|\alpha_j|^2$ and $|\beta_j|^2$ decay exponentially away from the concerned edge.

\begin{figure*}
\centering
    \subfloat[
    \label{fig:kitaev_1d_J1action_weak}
    ]
    {\includegraphics[width=.3\linewidth]{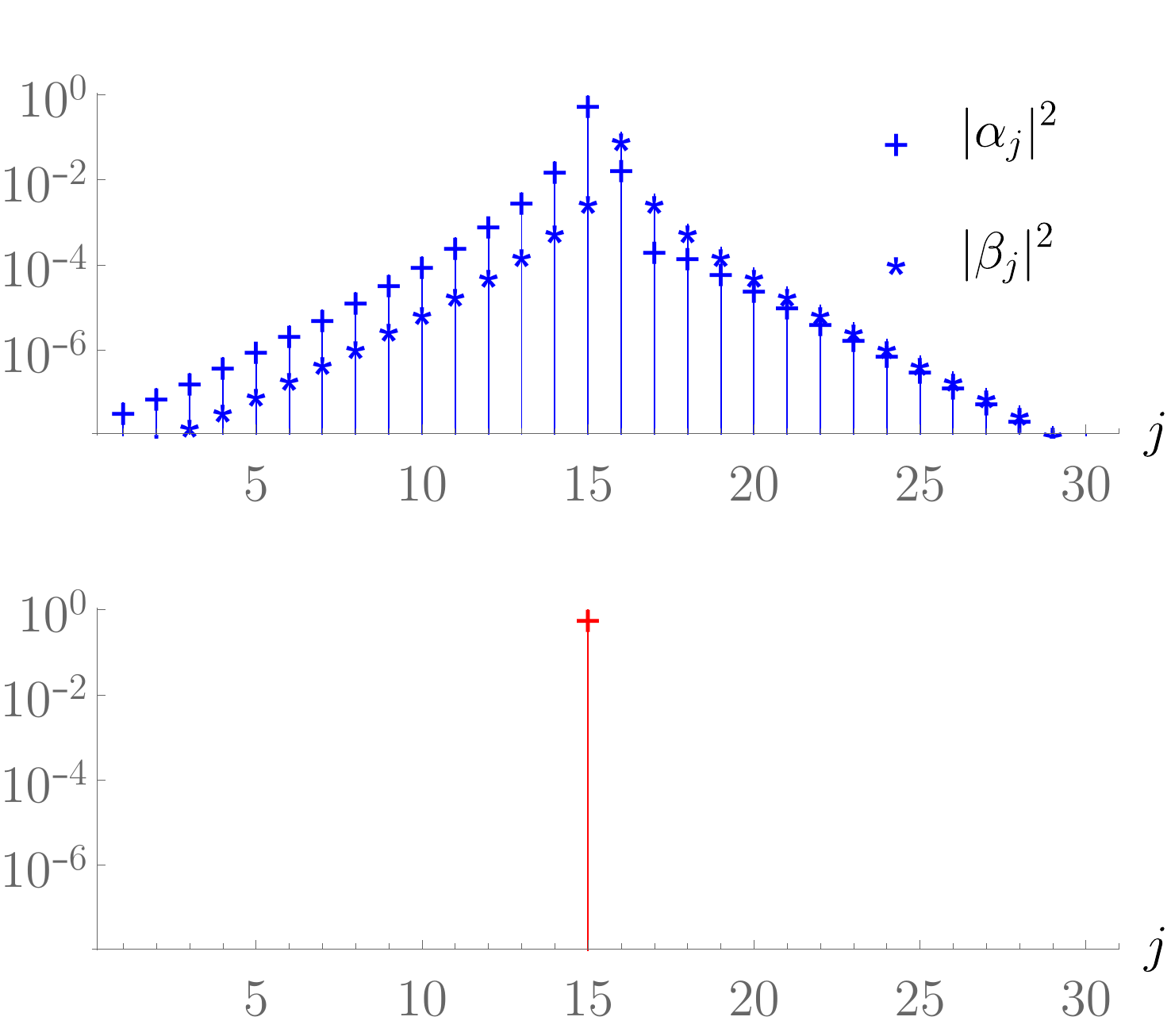}}
    \hspace{.03\linewidth}
    \subfloat[
    \label{fig:kitaev_1d_J1action_visualize}]{\includegraphics[width=.3\linewidth]
    {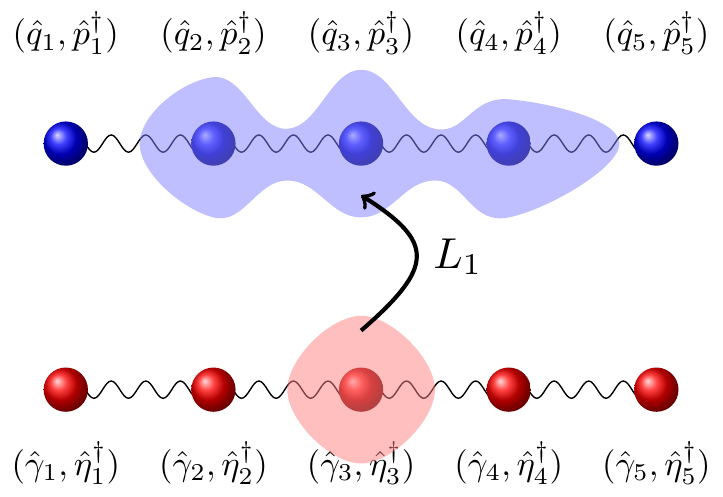}}
    \hspace{.03\linewidth}
    \subfloat[
    \label{fig:kitaev_1d_J1action_strong}]{\includegraphics[width=.3\linewidth]{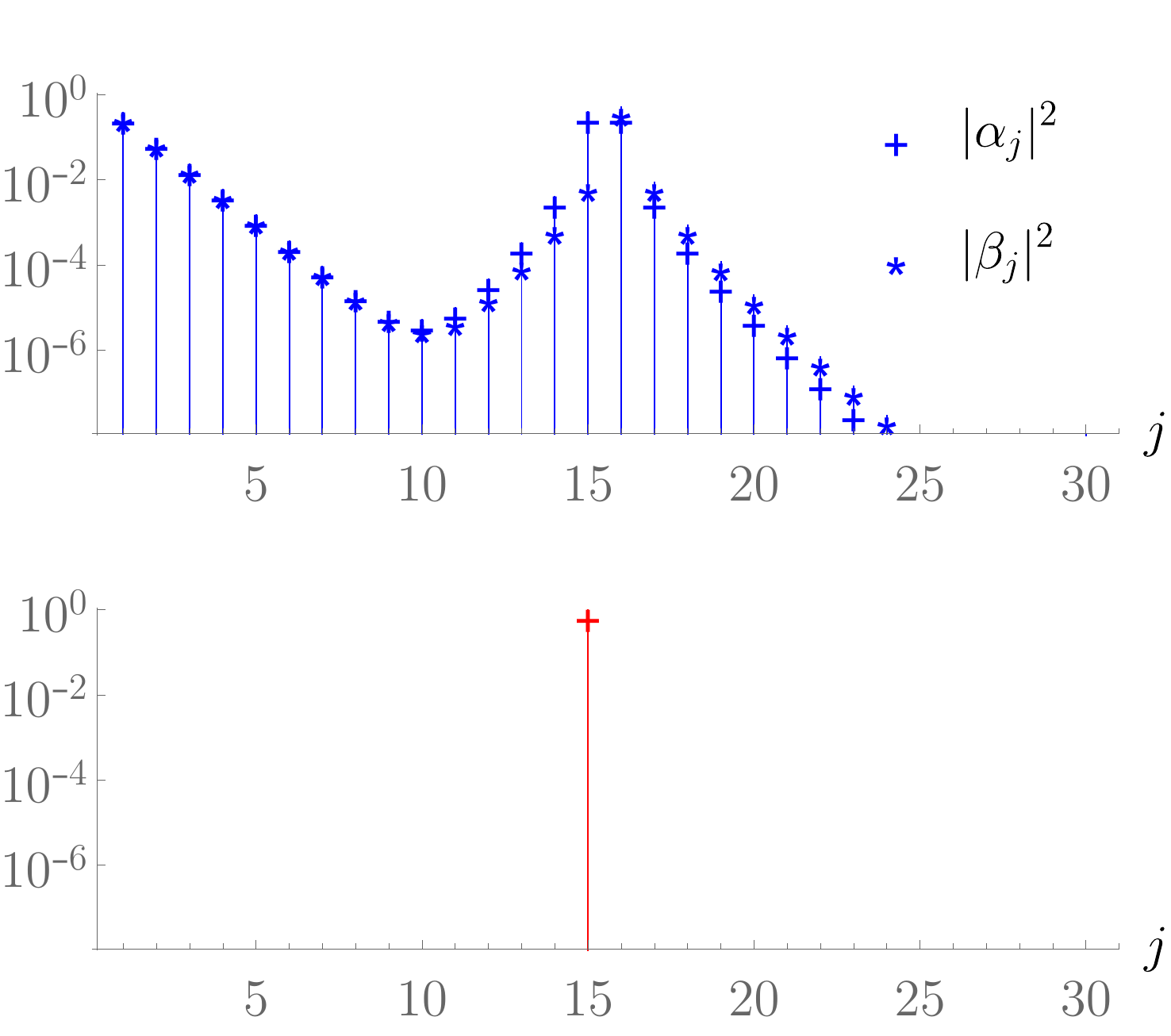}}
    \caption{Locality of the identification map $L_1$ and its dependence on the relative coupling $t/\mu$.
    As schematically visualized in (b),
    the plots show, for a system of $N=30$ modes, how $L_1$ associates the onsite operator $\hat c_{15}$ in the fermionic Kitaev chain \eqref{eq:Kitaev_SUSY_Hf} to the operator $L_1(\hat c_{15})=\sum_j\alpha_j\hat b_j+\beta_j\hat b_j^\dagger$ on the bosonic Kane-Lubensky chain \eqref{eq:Kitaev_SUSY_Hb}. 
    In the trivial phase, as plotted on the left in (a) for $t=0.35\mu$, the identification map preserves locality to a very high degree, namely, with an exponential decay of the coefficients $|\alpha_{j}|^2$ and $|\beta_{j}|^2$ with the distance $|k-j|$ (here $k\equiv 15$). 
    In the topological phase, as plotted on the right in (c) for $t=\mu$, the operator $L_1(\hat c_{k})$ can be non-local with a strong contribution from the boundary sites.
    }\label{fig:kitaev_1d_J1action}
\end{figure*}


The appearance and localization of the edge modes have consequences for the properties of the identification maps.
In particular, they affect to what extent the identification maps preserve the locality of the onsite observables in a system when mapping them onto its SUSY partner, as visualized in Fig.~\ref{fig:kitaev_1d_J1action}. 
From above, we know that the identification maps exactly map corresponding eigenmodes of the partner Hamiltonians to each other, and that we can choose the relative phase factor such that
\begin{align}\label{eq:kitaev_J1J2_choice}
    L_1(\hat{c}'_i)= \hat{b}'_i,\quad L_2(\hat{b}'_i)=-\ii \hat{c}'_i\,.
\end{align}
Thus, at the point ($t=0$), where the individual chain sites can be chosen as eigenmodes of the partner Hamiltonians, the identification maps exactly associate the fermionic and bosonic chain sites one-to-one,  maintaining their ordering. 

This feature of locality of the identification maps is conspicuous throughout the trivial phase, except the onsite localization at $t=0$ now transforms to an exponential one (with a length scale falling with the spectral gap), as seen in Fig.~\ref{fig:kitaev_1d_J1action_weak} for a chain of $N=30$ sites.
In detail, in the trivial phase, the identification maps associate single site operators $\hat c_i$ and $\hat b_i$ with operators $L_1(\hat c_k)=\sum_j \alpha_{kj}\hat b_j+\beta_{kj}\hat b_j$, such that $|\alpha_{kj}|^2$ and $|\beta_{kj}|^2$ decay exponentially in $|k-j|$. Likewise, in the trivial phase, $L_2$ maps onsite bosonic operators to exponentially localized fermionic operators.

In the topological phase, however, the identification maps develop non-local features, as can be seen in Fig.~\ref{fig:kitaev_1d_J1action_strong}. 
Here, a fermionic site operator $\hat c_k$ (\eg in the figure, $k=15$ in a chain of $N=30$ sites) when mapped to the operator $L_1(\hat c_{k})$ on the bosonic side, can acquire a significant component located at the left edge of the bosonic chain, which is the edge where also the bosonic edge mode is localized.
If we shift the original fermionic site to further right, the edge contribution to $L_1(\hat c_k)$  decays, and the localization of the resulting observable gains prominence. On the other hand, if we move the original site to the left, the edge contribution to $L_1(\hat c_k)$ starts dominating over the bulk coming from (bosonic) sites in the neighborhood of the $k$th site.

Instead of the map $L_1$, we may as well employ $L_2^{-1}$ to map the fermionic site operators to their bosonic counterparts. The observed behavior is similar; however, for $L_2^{-1}(\hat c_k)$, in the topological phase, the edge contribution at the left end of the bosonic chain dominates as $k\to N$, \ie when the original fermionic operator approaches the right end of the chain.

The converse association of bosonic onsite operators with the corresponding fermionic observables via the identification maps $L_2$ or $L_1^{-1}$ behaves very similar:
In the trivial phase, they are exponentially localized as above, and in the topological phase, they exhibit similar non-local features. 
However, here, both $L_2(\hat b_k)$ and  $L_1^{-1}(\hat b_k)$ develop a dominant edge contribution when the original bosonic operator $\hat b_k$ approaches the left end of the chain. For $L_2(\hat b_k)$, the edge contribution appears on the left edge of the fermionic chain; for $L_1^{-1}(\hat b_k)$, it appears on the right edge.

This example demonstrates that the identification between the bosonic and the fermionic parts of a SUSY Gaussian state via the identification maps may or may not coincide with an identification based intuitively on some underlying (lattice) geometry of the SUSY Hamiltonians.
Whereas we observe agreement in the trivial phase of the SUSY Kitaev chain, in the topological phase, the identification maps behave vastly  differently and disengage from notions based on the geometric intuition.
The duality relations of the next section will show that, whereas the geometrical appearance of modes can be distorted by the identification maps, their entanglement properties remain intimately related.

\section{Entanglement duality}\label{sec:entanglement-duality}
In this section, we derive how subsystem decompositions $V=A\oplus B$ behave under the supersymmetric identification maps $L_1$ and $L_2$, which leads to a duality between the bosonic and fermionic (mixed) Gaussian states. We can also use this to relate the associated entanglement entropies.

\subsection{Reduced Gaussian states and entanglement}\label{sec:entanglement_gaussian_states}
Given a classical phase space $V\simeq\mathbb{R}^{2N}$, a subspace $A\subset V$ defines a physical subsystem if the following condition is satisfied:
\begin{itemize}
    \item \textbf{Bosonic:} The restriction of $\Omega^{ab}$ to the subspace $A$ is non-degenerate, \ie has non-zero determinant.
    \item \textbf{Fermionic:} The subspace $A$ is even dimensional.
\end{itemize}
Note that the bosonic condition also implies that $A$ is even dimensional, as any anti-symmetric odd-dimensional matrix has a vanishing determinant.

In practice, we choose a basis $\hat{\xi}^a=(\hat{\xi}_A^a,\hat{\xi}^a_B)$ that splits $V=A\oplus B$ into a direct sum, where $B$ is the complementary system to $A$ defined as
\begin{equation}\label{eq:system_complement}
    \begin{aligned}
    B&=\begin{cases}
    \left\{v^a\in V\,\big|\, v^a \Omega^{-1}_{ab}u^b=0\,~\forall\, u^b\in A \right\}& \textbf{(bosons)}\\
    \left\{v^a\in V\,\big|\, v^a G^{-1}_{ab}u^b=0\,~\forall\, u^b\in A \right\}& \textbf{(fermions)}
    \end{cases}\,,
    \end{aligned}
\end{equation}
which is called the symplectic complement for bosons and the orthogonal complement for fermions.\footnote{Here, we used the inverse matrices $\Omega^{-1}$ and $G^{-1}$, which are bilinear forms on the phase space (rather than its dual). In~\cite{hackl2020geometry,hackl2020bosonic,windt2020local}, they are denoted by $\Omega^{-1}_{ab}\equiv \omega_{ab}$ and $G^{-1}_{ab}\equiv g_{ab}$.} 
We have the two bases $\hat{\xi}_A$ and $\hat{\xi}_B$ with $N_A+N_B=N$, such that the resulting matrix representations of $\Omega^{ab}$ and $G^{ab}$ take the forms
\begin{equation}
    \begin{aligned}
    \Omega^{ab}&\equiv\left(\begin{array}{cc|cc}
         & \id &&\\
    -\id     & &&\\
    \hline
    &&& \id\\
    &&-\id & 
    \end{array}\right) \equiv \left(\begin{array}{c|c}
        \Omega_A &  \\
        \hline
         & \Omega_B 
    \end{array}\right)\,,&&\textbf{(bosons)}\\
    G^{ab}&\equiv\left(\begin{array}{cc|cc}
    \id     &  &&\\
    & \id&&\\
    \hline
    &&\id& \\
    && & \id 
    \end{array}\right) \equiv \left(\begin{array}{c|c}
        G_A &  \\
        \hline
         & G_B 
    \end{array}\right)\,.&&\textbf{(fermions)}
    \end{aligned}
\end{equation}
Note that this implies that the restrictions $\Omega_A$ and $\Omega_B$, or $G_A$ and $G_B$, respectively, reproduce the standard forms from\ \eqref{eq:OmG-standard}, \ie the subsystems are themselves a bosonic or fermionic system consisting of $N_A$ and $N_B$ degrees of freedom.

When quantizing the subsystems $A$ and $B$, we can construct Fock spaces $\mathcal{H}_A$ and $\mathcal{H}_B$ as described in Sec.~\ref{eq:Gaussian-states}, such that the full Hilbert space is a tensor product $\mathcal{H}=\mathcal{H}_A\otimes\mathcal{H}_B$. In general, a pure Gaussian state $\ket{J}\in\mathcal{H}$ will itself not be a tensor product state with respect to this decomposition, which means that the subsystems are entangled.

It is well-known that the bipartite entanglement encoded in a general pure state $\ket{\psi}$ can be characterized by the spectrum of the mixed state $\rho_A=\Tr_{B}\ket{\psi}\bra{\psi}$ that results from tracing over $\mathcal{H}_B$. If $\ket{\psi}$ is a pure Gaussian state $\ket{J}$, the reduced state $\rho_A$ is a mixed Gaussian state. It can be expressed in terms of the linear complex structure $J$ as~\cite{hackl2020bosonic}
\begin{align}
    \rho_A=e^{-\hat{Q}}\,\,\,\text{with}\,\,\,\hat{Q}=\begin{cases}
    q_{rs}\hat{\xi}^r\hat{\xi}^s+c_0 & \normalfont{\textbf{(bosons)}}\\
    \ii q_{rs}\hat{\xi}^r\hat{\xi}^s+c_0 & \normalfont{\textbf{(fermions)}}
    \end{cases}\,,\label{eq:mixed_state}
\end{align}
with $q_{rs}$ is a  $2N_A$-by-$2N_A$ matrix\footnote{Not to confuse with the quadrature operator $\hat{q}_i$, which carries at most one index.} given by~\cite{hackl2020bosonic}
\begin{align}
\begin{split}
    q_{rs}&=\begin{cases}
	-\ii(\Omega^{-1}_A)_{rl}\,\mathrm{arccoth}\left(\ii J_A\right)^l{}_s & \textbf{(bosons)}\\[1mm]
	+\ii (G^{-1}_A)_{rl}\,\mathrm{arctanh}\left(\ii J_A\right)^l{}_s & \textbf{(fermions)}
	\end{cases}\\
	&=\begin{cases}
	+(\Omega^{-1}_A)_{rl}\,\mathrm{arccot}\left(J_A\right)^l{}_s & \textbf{(bosons)}\\[1mm]
	-(G^{-1}_A)_{rl}\,\mathrm{arctanh}\left(J_A\right)^l{}_s & \textbf{(fermions)}
	\end{cases}\,.\label{eq:J-q-relationship}
\end{split}
\end{align}
Here, $J_A$ is the restriction of $J$ to the $2N_A$-by-$2N_A$ subblock, representing the action of $J$ onto the subspace $A\subset V$, and similarly, $\Omega^{-1}_A$ and $G^{-1}_A$ denote the restrictions of $\Omega^{-1}$ and $G^{-1}$. The  coefficient $c_0$ is given by
\begin{align}
	c_0&=\begin{cases}
	\tfrac{1}{4}\log\det\left(\tfrac{\id+J^2_A}{4}\right) & \textbf{(bosons)}\\ -\tfrac{1}{4}\log\det\left(\tfrac{\id+J^2_A}{4}\right) & \textbf{(fermions)}
	\end{cases}\,.
\end{align}
It can be shown\ \cite{hackl2019minimal,hackl2020bosonic} that the eigenvalues of $J_A$ are purely imaginary and appear in $N_A$ conjugate pairs $\pm\ii\lambda_i$, where $\lambda_i\in[1,\infty)$ for bosons, and $\lambda_i\in[0,1]$ for fermions. 

These relations have the consequence that, for Gaussian states, the rather complicated spectrum of $\rho_A$ simplifies so that it can be efficiently calculated from the much simpler spectrum of $J_A$ given by $\pm\ii\lambda_i$.
Specifically, the eigenvalues of $\rho_A$ are
\begin{align}
   \hspace{-3mm}\mu(n_1,\dots,n_N)=\begin{cases}
    \left(\prod_{i=1}^{N_A}\frac{(\tanh r_i)^{n_i}}{\cosh r_i}\right)^2 & \textbf{(bosons)}\\
    \left(\prod_{i=1}^{N_A}\frac{(\tan{r_i})^{n_i}}{\sec{r_i}}\right)^2 & \textbf{(fermions)}
    \end{cases}\,,\label{eq:mu-spectra}
\end{align}
where $r_i=\frac{1}{2}\cosh^{-1}(\lambda_i)$, $n_i\in\mathbb{N}$ for bosons, and $r_i=\frac{1}{2}\cos^{-1}(\lambda_i)$, $n_i=0,1$ for fermions. 

The entanglement entropy $S_A(\ket{\psi})=S_B(\ket{\psi})$ is computed as the von Neumann entropy $S(\rho_A)$ of the reduced state $\rho_A$, namely,
\begin{align}
    S_A(\ket{\psi})=S(\rho_A)=-\Tr \rho_A\,\log{\rho_A}\,.
\end{align}
Calculating this quantity in practice is notoriously hard, as it requires computing the spectrum of $\rho_A$ that demands vast computational resources for large systems and appropriate approximations or truncation for infinite dimensional Hilbert spaces (in the case of bosons). However, if the state $\ket{\psi}$ happens to be a Gaussian state $\ket{J}$, we can exploit the relation between the spectra of $\rho_A$ and $J_A$ to find analytical formulas in terms of the restriction $J_A$ to the subsystem $A$. These restrictions correspond exactly to the symplectic or orthogonal decomposition $V=A\oplus B$ introduced at the beginning of this section. The formulas for the von Neumann entropies are given by\ \cite{sorkin1983entropy,peschel2003calculation} 
\begin{align}\label{eq:entropy_functions}
    &S(\rho_A)=\left\{\begin{array}{rl}
    \sum_{i=1}^{N_A} s_\bb(\lambda_i)& \textbf{(bosons)}\\
    \sum_{i=1}^{N_A} s_\ff(\lambda_i) & \textbf{(fermions)}
    \end{array}\right.\,,
\end{align}
with 
$s_{\bb}(x)=\left(\frac{x+1}{2}\right)\log\left(\frac{x+1}{2}\right)-\left(\frac{x-1}{2}\right)\log\left(\frac{x-1}{2}\right)$ for bosons, and 
$s_\ff(x)=-\frac{1+x}{2}\log\left(\frac{1+x}{2}\right)-\frac{1-x}{2}\log\left(\frac{1-x}{2}\right)$
for fermions,
which can be unified by the single trace formula\ \cite{bianchi2015entanglement,hackl2020bosonic}
\begin{align}
    S(\rho_A)=\frac{1}{2}\left|\Tr\left[\left(\frac{1+\ii J_A}{2}\right)\log\left(\frac{1+\ii J_A}{2}\right)^2\right]\right|\,.
\end{align}
The formula in~\eqref{eq:entropy_functions} can also be used to compute the Renyi entropy of order $n$ if we replace $s_\bb$ and $s_\ff$ by the respective Renyi entropy functions~\cite{hackl2020bosonic}:
\begin{align}\label{eq:renyi_formulas}
    r^{(k)}_\bb(\lambda)&=\frac{1}{k-1}\log\left((\tfrac{\lambda+1}{2})^k-(\tfrac{\lambda-1}{2})^k\right)\,,\\
    r^{(k)}_\ff(\lambda)&=-\frac{1}{k-1}\log\left((\tfrac{1+\lambda}{2})^k+(\tfrac{1-\lambda}{2})^k\right)\,.
\end{align}
It follows from the above that a subsystem $A$ (bosonic or fermionic) of a system in a pure Gaussian state is not entangled with the rest of the system, \ie it is in a product state with the rest of the system, if and only if $\lambda_i=1$ for all eigenvalues of $J_A$. In that case, we have $J_A^2=-\id_A$, and the subsystem is in a pure Gaussian state on its own. This is equivalent to $J(A)=A$, \ie  the full (unrestricted) linear complex structure mapping $A$ onto itself.

\subsection{Supersymmetric ground states and identification maps}
Above in~\eqref{eq:J1J2_JfJb}, we saw that $L_1$ and $L_2$ together encode the linear complex structures of both the bosonic and the fermionic part of the ground state \eqref{eq:groundstate_ket} of $\hat H$. 
In the following, we will use $L_1$ and $L_2$ to identify subsystems of fermionic modes with subsystems of bosonic modes, and vice versa. 

The maps $L_1$ and $L_2$ are the canonical choices for the identification maps because they preserve the Kähler structures of the fermionic ground state $\ket{J_\ff}$ and the bosonic ground  state $\ket{J_\bb}$.
That is, if we consider the fermionic 2-point function
\begin{align}
    C_{\ff,2}^{\alpha\beta}=\bra{J_\ff}\hat\xi^\alpha\hat\xi^\beta\ket{J_\ff}
    =\frac12\left(G_{\ff}^{\alpha\beta}+\ii\Omega_{\ff}^{\alpha\beta}\right)
\end{align}
and the bosonic
\begin{align}
    C_{\bb,2}^{ab}=\bra{J_\bb}\hat\xi^a\hat\xi^b\ket{J_\bb}
    =\frac12\left(G_{\bb}^{ab} +\ii\Omega_{\bb}^{ab}\right),
\end{align}
then one can show that we have
\begin{align}\label{eq:G_preserve}
    G_{\bb}^{ab}=(L_2)^a{}_\alpha G_{\ff}^{\alpha\beta} (L_2^\intercal)_\beta{}^b = (L_1^{-1})^a{}_\alpha G_{\ff}^{\alpha\beta} (L_1^{\intercal\, -1})_\beta{}^b\, ,
\end{align}
as well as (dropping the indices for a better readability)
\begin{align}\label{eq:Omega_preserve}
    \Omega_\ff= L_1\Omega_\bb L_1^\intercal= L_2^{-1}\Omega_\bb L_2^{-1\,\intercal}.
\end{align}
Thus,  the identification maps $L_1$ and $L_2$ preserve both the symmetric and the antisymmetric forms of the Kähler structure and exactly map the bosonic and fermionic 2-point functions of the ground state onto each other.
Interestingly, we  see that  it makes no difference  whether we use
$L_1$ and $(L_1)^{-1}$ or  $L_2$ and $(L_2)^{-1}$ for this purpose. The reason  is that both maps are closely related. In fact, since $J_\ff^2=-\id$ and $J_\bb^2=-\id$, it follows that
\begin{align}
\begin{split}
    (L_1)^{-1}&= -L_2J_\ff= -J_\bb L_2\,,\\
     (L_2)^{-1}&= -J_\ff L_1=-L_1J_\bb\,.
\end{split}
\end{align}

\subsection{Dual supersymmetric subsystems}

Since the identification maps $L_1$ and $L_2$  preserve the Kähler structures, subsystems in one part (bosonic/fermionic) of a supersymmetric Gaussian state can be identified with subsystems in the other part (fermionic/bosonic).

If $A\subset V_\bb$ corresponds to a bosonic subsystem, then both $L_1(A)$ and $L_2^{-1}(A)$ are even-dimensional subspaces of the fermionic phase space $V_\ff$; hence, they correspond to a fermionic subsystem, as defined in Sec.~\ref{sec:entanglement_gaussian_states}.
If, on the other hand, $A\subset V_\ff$ corresponds to a fermionic subsystem, then $L_2(A)$ and $L_1^{-1}(A)$ only correspond to a bosonic subsystem, if the restriction of $\Omega_\ff$ to $A$ is non-degenerate. Following \eqref{eq:Omega_preserve}, this condition ensures that $\Omega_\bb$ is non-degenerate as required for $L_2(A)$ and $L_1^{-1}(A)$ to yield a bosonic subsystem.

How does the subsystem which $A$ is mapped to depend on whether we use the identification map $L_1$ (and its inverse) or the map $L_2$?
If the subsystem $A$ is in a pure state, there is no difference; both identification maps identify $A$ with the same subsystem. For example, if $A\subset V_\bb$ is a bosonic subsystem which is in a pure state, then we have $J_\bb(A)=A$; thus,
\begin{align}
    L_1(A)=L_1(J_\bb(A))=L_2^{-1}(A).
\end{align}

However, for an entangled subsystem, we have $J(A)\neq A$ and are led to the following commutative diagram:
\begin{equation}\label{eq:comm_diagramm_subsystems}
\begin{tikzcd}
A_\bb \arrow[r, "L_1"] \arrow[leftrightarrow]{d}[left]{J_\bb} \arrow[leftarrow]{dr}[near start]{L_2} & A_\ff \arrow[d, "J_\ff", leftrightarrow] \arrow{dl}[near end]{L_2} \\ \tilde A_\bb \arrow{r}[below]{L_1} & \tilde A_\ff   
\end{tikzcd}
\end{equation}
Here, we have chosen  $A_\bb \subset V_\bb$ as a bosonic subsystem, defined $A_\ff=L_1(A_\bb)$ and denoted $\tilde A_\bb=J_\bb(A_\bb)$ and $\tilde A_\ff=J_\ff(A_\ff)$.

Whereas $A_\bb$ and $\tilde A_\bb$ ($A_\ff$ and $\tilde A_\ff$) define different bosonic (fermionic) subsystems, they are intimately related: $A_\bb \cup \tilde A_\bb$ ($A_\ff \cup \tilde A_\ff$) is the smallest subsystem containing $A_\bb$ ($A_\ff$) which is in a pure partial state,  \ie  shares no entanglement with the rest of the system.

Furthermore, $A_\bb$ ($A_\ff$) shares the same amount of entanglement with the rest of the system as does  $\tilde A_\bb$ ($\tilde A_\ff$). This follows from the fact that the restricted linear complex structures $J^\bb_{A_\bb}$ and $J^\bb_{\tilde A_\bb}$ ($J^\ff_{A_\ff}$ and $J^\ff_{\tilde A_\ff}$) have the same spectrum.
To see this, consider the decomposition of the phase space into the direct sum $V_\bb=A_\bb\oplus B_\bb$ according to \eqref{eq:system_complement}. We define by $P_{A_\bb}$ the projector onto $A$ with respect to this decomposition:
\begin{align}
P_{A_\bb}(A_\bb)=A_\bb,\quad P_{A_\bb}(B_\bb)=0.
\end{align}
The restriction of $J_\bb$ to $A_\bb$ is then
$
J^\bb_{A_\bb}=P_{A_\bb} J_\bb P_{A_\bb}.
$
Analogously, considering the decomposition $V_\bb=\tilde A_\bb\oplus \tilde B_\bb$, we find that the projector onto $\tilde A_\bb$ is
$
P_{\tilde A_\bb}= - J_\bb P_{A_\bb} J_\bb,
$
and 
\begin{align}
    J^\bb_{\tilde A_\bb}= P_{\tilde A_\bb}J_\bb P_{\tilde A_\bb}=-J_\bb J^\bb_{A_\bb} J_\bb\,.
\end{align}
Since  $J_\bb^{-1}=-J_\bb$, $J^\bb_{\tilde A_\bb}$ and $J^\bb_{\tilde A_\bb}$ are represented by similar matrices and, hence, have the same spectrum.
In fact,  if $v\in A_\bb$ is an eigenvector of $J^\bb_{A_\bb}$ with $J^\bb_{A_\bb} v=\pm \ii \lambda v$, then $J_\bb v$ is an eigenvector of $J^\bb_{\tilde A_\bb}$ with the same eigenvalue.

\subsection{Duality for Gaussian states and their entanglement}
In the previous section, we analyzed the structure of subsystems in supersymmetric Gaussian states. In particular, we discussed how the identification maps $L_1$ and $L_2$ relate bosonic subsystems to fermionic subsystems and vice versa. We can now use this background structure to derive the following duality between bosonic and fermionic Gaussian states.

The setting is as follows. We consider a classical phase space $V\simeq \mathbb{R}^{2N}$ with Kähler compatible structures $(G,\Omega,J)$ and a choice of a subspace $A\subset V$ with $\dim A=2N_A$. We can associate two distinct quantum theories, namely, a bosonic Hilbert space $\mathcal{H}_\bb$ with Gaussian state $\ket{J}_\bb$ and a fermionic Hilbert space $\mathcal{H}_\ff$ with Gaussian state $\ket{J}_\ff$. In both quantum theories, we can construct a reduced density operator $\rho_A$ whose spectrum is determined by the restricted complex structure.

Crucially, however, the restriction of $J$ to $A$ is different depending on whether we consider a bosonic system and use a symplectic decomposition of the phase space or consider a fermionic system and use an orthogonal decomposition, according to \eqref{eq:system_complement}.
This is due to the fact that the $2N_A$-by-$2N_A$ subblock of the matrix $J$ associated to the subspace $A$ depends also on the basis elements that are not contained in $A$. In particular, we choose two different bases for the bosonic and fermionic case $\hat{\xi}_\bb=(\hat{\xi}^A_\bb,\hat{\xi}^B_\bb)$ and $\hat{\xi}_\ff=(\hat{\xi}^A_\ff,\hat{\xi}^B_\ff)$, such that
\begin{align}
    \mathrm{span}(\hat{\xi}^A_\bb)&=A=\mathrm{span}(\hat{\xi}^A_\ff)\,,\\
    \mathrm{span}(\hat{\xi}^B_\bb)&=B_\bb\neq B_\ff=\mathrm{span}(\hat{\xi}^B_\ff)\,,
\end{align}
where $B_\bb$ and $B_\ff$ are the respective bosonic and fermionic complements defined in~\eqref{eq:system_complement}. 
Consequently, the restrictions of $J$ to the subspace $A$ can be different on the bosonic and the fermionic side, which therefore are denoted by $J_A^\bb$ and $J_A^\ff$, respectively.
Equipped with this, we can now prove the following proposition.

\begin{proposition}[Entanglement duality]\label{prop:entanglement-duality}
We consider a supersymmetric system with phase space $V\simeq V_\bb\simeq V_\ff$ equipped with Kähler structures $(G,\Omega,J)$, which simultaneously describe a bosonic and a fermionic Gaussian state, namely, $\ket{J}_\bb\in\mathcal{H}_\bb$ and $\ket{J}_\ff\in\mathcal{H}_\ff$. We now choose a subsystem $A\subset V$. This leads to two inequivalent decompositions of $V$, namely, $V=A\oplus B_\bb$ and $V=A\oplus B_\ff$, where the complementary subsystems $B_\bb$ and $B_\ff$ are defined in~\eqref{eq:system_complement}. The associated reduced states $\rho_A^\bb$ and $\rho_A^\ff$ are both Gaussian and fully described by the restricted complex structure $J_A^\bb$ and $J_A^\ff$, respectively, which satisfy the following relation:
\begin{align}
    J_A^\ff=-(J_A^\bb)^{-1}\,.\label{eq:entanglement-duality}
\end{align}
This implies that the eigenvalues $\pm\ii\lambda_i^\bb$ of $J_A^\bb$ are related to the eigenvalues $\pm\ii\lambda_i^\ff$ of $J_A^\ff$ via $\lambda_i^\bb=1/\lambda^\ff_i$.
\end{proposition}
\begin{proof}
The decompositions $V=A\oplus B_\bb$ and $V=A\oplus B_\ff$ define projectors, such that $P_\bb: V\to A$, $P_\ff: V\to A$, $\bar{P}_\bb: V\to B_\bb$, and $\bar{P}_\ff: V\to A_\ff$, such that $\id=P_\bb+\bar{P}_\bb=P_\ff+\bar{P}_\ff$. The restricted complex structures are then defined as
\begin{align}
    J_A^\bb&=P_\bb J|_A: A\to A\,,\\
    J_A^\ff&=P_\ff J|_A: A\to A\,.
\end{align}
We need to show $J_A^\ff=-(J_A^\bb)^{-1}$, which is equivalent to $J_A^\ff J_A^\bb=-\id_A$.
To show the latter, we take a vector $a\in A$ and calculate
\begin{align}
\begin{split}
    -a &= -P_\ff a=P_\ff J^2 a  =P_\ff J (P_\bb+\bar{P}_\bb) J  a \\
    &= P_\ff J J^\bb_A a + P_\ff J \bar{P}_\bb J a= J_A^\ff J_A^\bb a+  P_\ff J \bar{P}_\bb J a\,.\label{eq:relation1}
\end{split}
\end{align}
The second term in~\eqref{eq:relation1} vanishes since, for an arbitrary vector $v\in V$, the inner product
\begin{align}
\begin{split}
    G^{-1}(v,P_\ff J\bar{P}_\bb a)&=G^{-1}(P_\ff v,J\bar{P}_\bb Ja)\\
    &=-\Omega^{-1}(\underbrace{P_\ff v}_{\in A},\underbrace{\bar{P}_\bb Ja}_{B_\bb})=0\,,\label{eq:zero}
\end{split}
\end{align}
where we have used the relationship $G^{-1}(\cdot,J\cdot)=-\Omega^{-1}(\cdot, \cdot)$ following from~\eqref{eq:GOmJ-relations}. In matrix notation, we would write $G^{-1}(v,w)=(G^{-1})_{ab}v^aw^b$ and so on. That the inner product $\Omega^{-1}(\cdot, \cdot)$ in \eqref{eq:zero} vanishes follows from the definition of $B_\bb$ in~\eqref{eq:system_complement}, and therefore, proves the identity in \eqref{eq:entanglement-duality}. 
\end{proof}

At first glance, this result is a simple statement about restricting a complex structure $J: V\to V$ to a subspace $A\subset V$ in two inequivalent ways. However, its application to bosonic and fermionic Gaussian states implies a rather complicated relationship of the spectra $\rho_A^\bb$ and $\rho_A^\ff$ via~\eqref{eq:mu-spectra} and~\eqref{eq:entanglement-duality}, which can be made precise in the following corollary relating the restricted complex structures of the dual subsystems.

\begin{corollary}

Given a supersymmetric system with supercharge operator $\hat{\mathcal{Q}}$, we have a supersymmetric ground state $\ket{J_\bb}\otimes\ket{J_\ff}$ of $\hat{H}=\hat{\mathcal{Q}}^2$ and identification maps $L_1: V_\bb\to V_\ff$, $L_2: V_\ff\to V_\bb$, and their inverses $L_1^{-1}$ and $L_2^{-1}$, as above. Then Proposition~\ref{prop:entanglement-duality} implies the following.\\
Let $S\subset V_\bb$ be a bosonic subsystem and $L(S)\subset V_\ff$, with $L=L_1$ or $L=L_2^{-1}$, be a dual fermionic subsystem. Then the restricted linear structures $J^\bb_S:S\to S$ and $J^\ff_{L(S)}:L(S)\to L(S)$ of these two subsystems are such that
\begin{align}
\begin{split}
   (J^\bb_S)^{-1} &= - L^{-1} J^\ff_{L(S)} L\,, \\
    (J^\ff_{L(S)})^{-1} &= - LJ^\bb_S L^{-1}\, .
\end{split}
\end{align}
Let $R\subset V_\ff$ be a fermionic subsystem and $L(R)\subset V_\bb$, with $L=L_2$ or $L=L_1^{-1}$, be a dual bosonic subsystem.
Then the restricted linear structures $J^\ff_R:R\to R$ and $J^\bb_{L(R)}:L(R)\to L(R)$ of these two subsystems are such that
\begin{align}
\begin{split}
    (J^\ff_R)^{-1} &= - L^{-1} J^\bb_{L(R)} L\,, \\ 
   ( J^\bb_{L(R)})^{-1} &= - L J^\ff_R L^{-1}\, .
\end{split}
\end{align}
Thus, the eigenvalues of the dual restricted complex structures are inverses of each other  as implied by~\eqref{eq:entanglement-duality}, and their entanglement spectra are accordingly related by~\eqref{eq:mu-spectra}. 
\end{corollary}

While our result applies to any identification where a bosonic and a fermionic phase space are related, supersymmetric systems with the identification maps $L_1$ and $L_2$, as discussed in Sec.~\ref{sec:susy-identification}, are the prime examples where such an identification is naturally chosen. 

The entanglement duality implies an intimate relation of the entanglement entropy of a subsystem with that of its dual subsystem because both the von Neumann entropy~\eqref{eq:entropy_functions}, as well as the Renyi entropies~\eqref{eq:renyi_formulas} are functions of the spectrum of the restricted complex linear structure.
For the simplest possible case, where the subsystems each consist of a single mode only, Fig.~\ref{fig:sf_sb_dual} shows the relation between the von Neumann entropy of the fermionic mode and the bosonic mode. Here, the restricted complex structures have one pair of imaginary eigenvalues, $\pm\ii\lambda$ for the fermionic   and $\pm\ii \lambda^{-1}$ for the bosonic system, which with the formula for the von Neumann entropies~\eqref{eq:entropy_functions} yields the relation plotted in Fig.~\ref{fig:sf_sb_dual}.

Evidently, the bosonic and the fermionic entanglement become asymptotically equal when the corresponding modes approach a pure partial state, and consequently, the entanglement approaches zero ($\lambda \rightarrow 1$). In the opposite direction, however, the entanglement in the bosonic mode grows without a bound as $\lambda\to0$, whereas the entanglement in the dual fermionic mode tends to saturate at the maximal value of $\log{2}$.

\begin{figure}
\begin{center}
  \includegraphics[width=.9\linewidth]{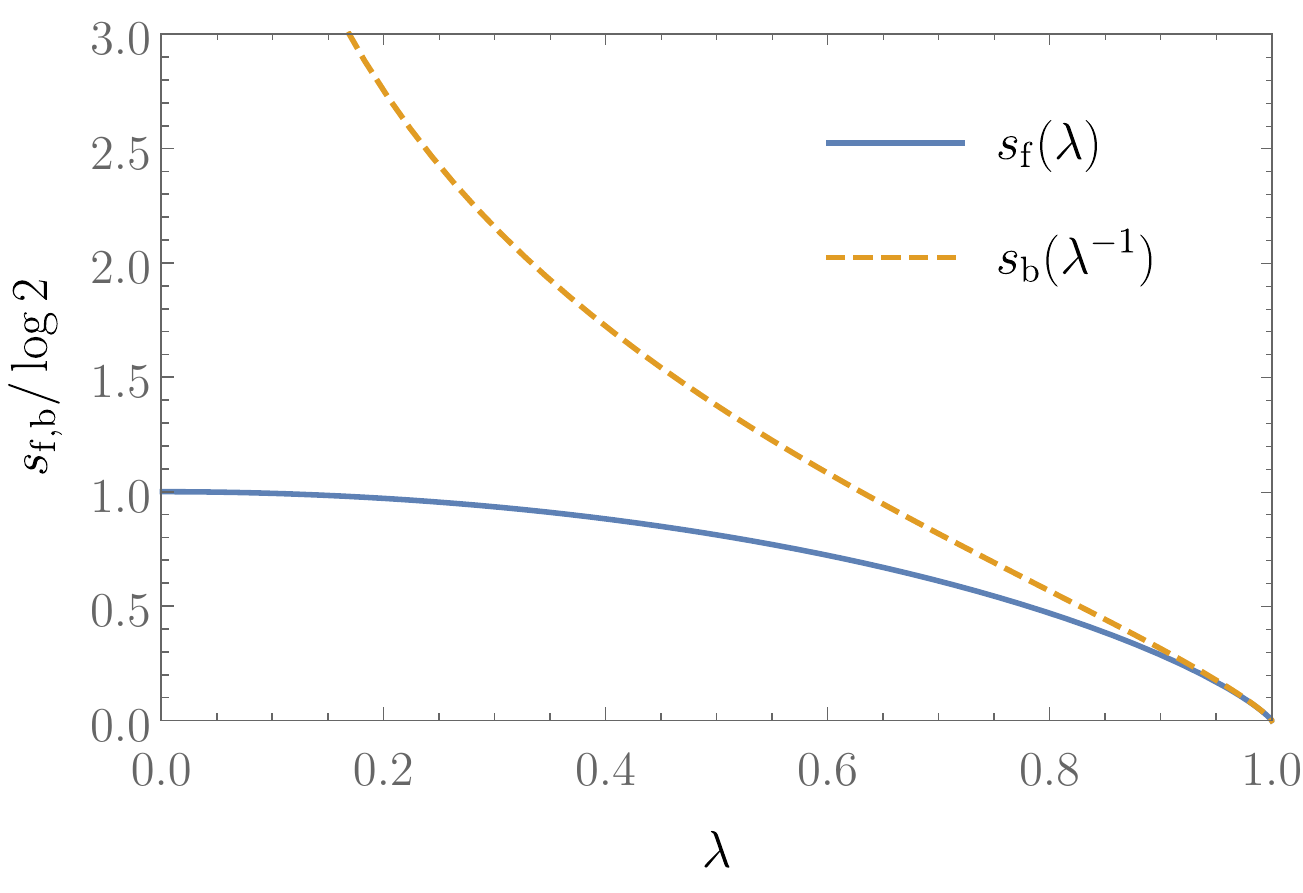}
\end{center}\vspace{-5mm}
\caption{Entanglement entropies for dual subsystems each consisting of a single mode.
The solid line shows the entanglement entropy of a single fermionic mode for which the restricted complex structure has eigenvalues $\pm\ii\lambda$.
Due to the entanglement duality \eqref{eq:entanglement-duality}, the restricted  complex structure of the dual bosonic mode has eigenvalues $\pm\ii\lambda^{-1}$, and the dashed line plots the resulting entanglement entropy.
} 
\label{fig:sf_sb_dual} 
\end{figure}

This relation between the SUSY partner single modes readily extends to multiple modes because, as is evident from \eqref{eq:entropy_functions}, the total entanglement entropy of a subsystem is given by the sum of the entanglement entropies over the individual normal modes of that subsystem. This is related to the fact, that a mixed Gaussian state always can be expressed as the product state of its normal modes, which are given by the eigenmodes of the restricted linear complex structure \cite{hackl2019minimal}. As a consequence of the entanglement duality, the identification maps identify normal modes with reciprocal eigenvalues of the restricted complex linear structures.

At this stage, it is an important question whether the entanglement duality is merely an interesting observation or to what extent  it matters for physical systems. 
In the following two sections, we therefore  investigate the entanglement duality in two concrete applications. First, we consider the toy model of a supersymmetric system with two bosonic and two fermionic modes in Sec.~\ref{sec:two_mode_duality} before we then move on to the recently proposed SUSY Kitaev honeycomb model in Sec.~\ref{sec:Kitaev_2D}.

Before proceeding, we note that mixed states that arise as a thermal state of a supersymmetric Hamiltonian also come under the ambit of our duality. In detail, we ask if a relation, like what applies to the mixed states arising from a reduction of a pure state to a subsystem $A$, holds for a thermal state of a supersymmetric Hamiltonian $\hat{H}=\hat{\mathcal{Q}}^2=\hat{H}_\bb+\hat{H}_\ff$, \ie $\rho=\frac{1}{Z}e^{\beta\hat{H}}=\frac{1}{Z}e^{\beta\hat{H}_\bb}\otimes e^{\beta\hat{H}_\ff}$. The following proposition answers this question in the affirmative.

\begin{proposition}[Thermal state duality]\label{prop:thermal-duality}
We consider a supersymmetric system with phase space $V\simeq V_\bb\simeq V_\ff$ equipped with Kähler structures $(G,\Omega,J)$, for which we have a Hamiltonian $\hat{H}=\hat{\mathcal{Q}}^2=\hat{H}_\bb+\hat{H}_\ff$. The thermal state $\rho=\frac{1}{Z}e^{\beta\hat{H}}$ at inverse temperature $\beta$ is a tensor product of two Gaussian states $\rho=\rho_\bb\otimes\rho_\ff$ with associated restricted complex structures $J_\bb$ and $J_\ff$ related by
\begin{align}
    J_\bb=-(J_\ff)^{-1}\,,\label{eq:thermal-state-duality}
\end{align}
which exactly resembles the entanglement duality, but now applies to the whole system. This implies that the eigenvalues $\pm\ii\lambda_i^\bb$ of $J_\bb$ are related to the eigenvalues $\pm\ii\lambda_i^\ff$ of $J_\ff$ via $\lambda_i^\bb=1/\lambda^\ff_i$.
\end{proposition}
\begin{proof}
Our identification of the phase spaces $V\simeq V_\bb\simeq V_\ff$ gives rise to a single Lie algebra generator $K: V\to V$ for the Hamiltonian $\beta\hat{H}$ from~\eqref{eq:sym-hamiltonian}. The spectrum of $K$ agrees with that of the bosonic generator $K_\bb: V_\bb\to V_\bb$ as well as the fermionic generator $K_\ff: V_\ff\to V_\ff$ defined as $(K_\bb)^a{}_b=\beta\,\Omega^{ac}h^\bb_{cb}$ and $(K_\ff)^\alpha{}_\delta=\beta\, G^{\alpha\gamma}q^\ff_{\gamma\delta}$, respectively. We can compare with~\eqref{eq:mixed_state} to identify that $\rho=e^{-\beta\hat{H}}/Z$ gives rise to $q^\bb_{ab}=\frac{\beta}{2}h^\bb_{ab}=\frac{\beta}{2}\Omega^{-1}_{ac}(K_\bb)^c{}_b$ and $q^\ff_{\alpha\gamma}=\frac{\beta}{2}h^\ff_{\alpha\gamma}=\frac{\beta}{2}G^{-1}_{\alpha\delta}(K_\ff)^\delta{}_\gamma$. We can invert~\eqref{eq:J-q-relationship} to find
\begin{align}
    J_\bb&=-\cot{\Omega q^\bb}=-\cot(K_\bb/2)\equiv-\cot(K/2)\,,\\
    J_\ff&=\tan{G q^\ff}=\tan(K_\ff/2)\equiv\tan(K/2)\,,
\end{align}
from which~\eqref{eq:thermal-state-duality} readily follows.
\end{proof}

\subsection{Application: two-mode system}\label{sec:two_mode_duality}
In this section, we study some consequences of the entanglement duality in a basic two-mode example where the SUSY Hamiltonian is given by a fermionic and a bosonic two-mode squeezing Hamiltonian. While this is a minimal example, it explains certain basic relations which are important for our analysis of a lattice Hamiltonian in the next subsection.
In particular, it highlights that almost maximally entangled fermionic modes are mapped to almost degenerate bosonic modes. This relation is central to the entanglement scaling in 2D systems discussed in Sec.~\ref{sec:Kitaev_2D}.

Consider the following supercharge operator $\hat{\mathcal{Q}}$, which is parametrized by real numbers $r_\bb\geq0$ and $0\leq r_\ff < \pi/4$, corresponding to squeezing parameters:
\begin{align}
\hat{\mathcal{Q}}=&\left( \cosh(r_\bb)\cos(r_\ff)-\sinh(r_\bb)\sin(r_\ff)\right) \left(\hat \gamma_1\hat q_1+\hat\eta_1\hat p_1\right)
\nn&+\left( \cosh(r_\bb)\cos(r_\ff)+\sinh(r_\bb)\sin(r_\ff)\right) \left(\hat \gamma_2\hat q_2+\hat\eta_2\hat p_2\right)
\nn&+\left(\cosh(r_\bb)\sin(r_\ff)-\sinh(r_\bb)\cos(r_\ff)\right)\left(\hat\gamma_1\hat q_2-\hat\eta_1\hat p_2\right)
\nn&+\left(\cosh(r_\bb)\sin(r_\ff)+\sinh(r_\bb)\cos(r_\ff)\right) \left( -\hat\gamma_2\hat q_1+\eta_2\hat p_1\right)\,.
\end{align}
It generates a SUSY Hamiltonian $\hat{H}=\hat{H}_\bb+\hat{H}_\ff$ which consists of the two-mode Hamiltonians:
\begin{align}
\hat{H}_\bb &=\frac{\cosh(2r_\bb)}2\sum_{i=1,2}\left(\hat q_i^2+\hat p_i^2\right)+ \sinh(2r_\bb)\left(\hat p_1\hat p_2-\hat q_1\hat q_2 \right)\,,\\
\hat{H}_\ff&=\frac{\ii \cos(2r_\ff)}2\sum_{i=1,2}\left(\hat \gamma_i \hat\eta_i-\hat \eta_i\hat\gamma_i\right)+\ii\sin(2r_\ff)\left(\hat\gamma_1\eta_2-\gamma_2\eta_1\right)\,.
\end{align}
The ground states of these Hamiltonians are two-mode squeezed states. Accordingly, the identification maps $L_1$ and $L_2$ are represented by
\begin{widetext}
\begin{align}
\begin{split}
\tiny \hspace{-2mm} 
L_1
    \equiv\left(
    \begin{array}{cccc}
    \cos(r_\ff) \cosh (r_\bb)-\sin(r_\ff) \sinh (r_\bb) & \sin(r_\ff) \cosh (r_\bb)-\cos(r_\ff) \sinh (r_\bb) & 0 & 0 \\
    -\sin(r_\ff) \cosh (r_\bb)-\cos(r_\ff) \sinh (r_\bb) & \sin(r_\ff) \sinh (r_\bb)+\cos(r_\ff) \cosh (r_\bb) & 0 & 0 \\
    0 & 0 & \cos(r_\ff) \cosh (r_\bb)-\sin(r_\ff) \sinh (r_\bb) & \cos(r_\ff) \sinh (r_\bb)-\sin(r_\ff) \cosh (r_\bb) \\
    0 & 0 & \cos(r_\ff) \sinh (r_\bb)+\sin(r_\ff) \cosh (r_\bb) & \sin(r_\ff) \sinh (r_\bb)+\cos(r_\ff) \cosh (r_\bb) \\
    \end{array}
    \right),\\
\tiny \hspace{-2mm}
L_2
    \equiv\left(
    \begin{array}{cccc}
    0 & 0 & \cos(r_\ff) \cosh (r_\bb)-\sin(r_\ff) \sinh (r_\bb) & \cos(r_\ff) \sinh (r_\bb)+\sin(r_\ff) \cosh (r_\bb) \\
    0 & 0 & \cos(r_\ff) \sinh (r_\bb)-\sin(r_\ff) \cosh (r_\bb) & \sin(r_\ff) \sinh (r_\bb)+\cos(r_\ff) \cosh (r_\bb) \\
    \sin(r_\ff) \sinh (r_\bb)-\cos(r_\ff) \cosh (r_\bb) & \cos(r_\ff) \sinh (r_\bb)+\sin(r_\ff) \cosh (r_\bb) & 0 & 0 \\
    \cos(r_\ff) \sinh (r_\bb)-\sin(r_\ff) \cosh (r_\bb) & -\sin(r_\ff) \sinh (r_\bb)-\cos(r_\ff) \cosh (r_\bb) & 0 & 0 \\
    \end{array}
    \right).
\end{split}
\label{eq:two_mode_J1J2}
\end{align}
\end{widetext} 
They lead to the complex structures
\begin{align}
J^\bb
&\equiv\left(
\begin{array}{cccc}
 0 & 0 & \cosh (2 r_\bb) & \sinh (2 r_\bb) \\
 0 & 0 & \sinh (2 r_\bb) & \cosh (2 r_\bb) \\
 -\cosh (2 r_\bb) & \sinh (2 r_\bb) & 0 & 0 \\
 \sinh (2 r_\bb) & -\cosh (2 r_\bb) & 0 & 0 \\
\end{array}
\right),
\end{align}
 \begin{align}
J^\ff
&\equiv\left(
\begin{array}{cccc}
 0 & 0 & \cos (2 r_\ff) & \sin (2 r_\ff) \\
 0 & 0 & -\sin (2 r_\ff) & \cos (2 r_\ff) \\
 -\cos (2 r_\ff) & \sin (2 r_\ff) & 0 & 0 \\
 -\sin (2 r_\ff) & -\cos (2 r_\ff) & 0 & 0 \\
\end{array}
\right),
\end{align}
which define pure two-mode squeezed states.

Let us now study how the identification maps act on the single site modes $(\hat q_1,\hat p_1)$ and $(\hat \gamma_1,\hat \eta_1)$, respectively. It is clear from \eqref{eq:two_mode_J1J2} that,  when both the bosonic and fermionic squeezing vanish, \ie $r_\bb=0=r_\ff$, these are trivially identified with each other. However, when either squeezing parameter takes a non-zero value, the identification maps will mix the modes $1$ and $2$.

Beginning with the bosonic mode $S=(\hat q_1,\hat p_1)$, we find it has the restricted complex linear structure 
\begin{align}
J_{S}^\bb\equiv \begin{pmatrix}0&\cosh(2r_\bb)\\ -\cosh(2r_\bb)& 0\end{pmatrix}\,,
\end{align}
which has eigenvalues $\pm\ii \cosh(2r_\bb)$, signaling that, for $r_\bb>0$, the mode is in a mixed Gaussian state due to its entanglement with mode $2$. Now we can use $L_2$ to associate this bosonic mode with a fermionic mode. Here, we need to consider that  the fermionic observables $L_2(\hat q_1)$ and $L_2(\hat p_1)$ are not properly normalized Majorana operators. In fact, as a consequence of \eqref{eq:G_preserve}, we have
\begin{align}
\bra{J_\ff} \left\{L_2(\hat q_1),L_2(\hat p_1)\right\}\ket{J_\ff}
=\bra{J_\bb} \left\{ \hat q_1, \hat p_1\right\}\ket{J_\bb}=\cosh(2r_\bb).
\end{align}
Instead, the properly normalized Majorana operators, which correspond to an orthogonal basis in the fermionic phase space, are 
\begin{align}
\tilde L_2(\hat q_1) :=\frac{L_2(\hat q_1)}{\sqrt{\cosh(2r_\bb)}},\quad \tilde L_2(\hat p_1) :=\frac{L_2(\hat p_1)}{\sqrt{\cosh(2r_\bb)}}\,,
\end{align}
which we can use to calculate the restriction of $J_\ff$ to this subsystem, \eg  to calculate its entanglement with the rest of the fermionic system. When calculating their commutator, we can make use of~\eqref{eq:Omega_preserve} to find
\begin{align}
&\bra{J_\ff}\tilde L_2(\hat q_1)\tilde L_2(\hat p_1)-\tilde L_2(\hat p_1)\tilde L_2(\hat q_1)\ket{J_\ff}
\nn&= \frac1{\cosh(2r_\bb)}\bra{J_\bb} \hat q_1 \hat p_1- \hat p_1 \hat q_1\ket{J_\bb}=\frac\ii{\cosh(2r_\bb)}.
\end{align}
Hence, the restriction of $J_\ff$ has eigenvalues $\pm\ii(\cosh(2r_\bb))^{-1}$, as predicted by the entanglement duality.

The opposite identification of the fermionic mode $R=(\hat{\gamma}_1,\hat{\eta}_1)$ with a bosonic mode is completely analogous, but it additionally highlights an important effect in the limit of $r_\ff\to\pi/4$. The restriction of $J_\ff$ to $(\hat\gamma_1,\hat\eta_1)$  is represented by
\begin{align}
    J^\ff_R\equiv\begin{pmatrix} 0&\cos(2r_\ff)\\-\cos(2r_\ff)&0\end{pmatrix},
\end{align}
which has eigenvalues $\pm\ii\cos(2r_\ff)$. Mapping $\hat\gamma_1$ and $\hat\eta_1$ via the identification map $L_1$, we obtain two bosonic observables which obey
\begin{align}\label{eq:commutator_on_bosonmode}
    \bra{J_\bb}\comm{L_1(\hat \gamma_1)}{L_1(\hat \eta_1)}\ket{J_\bb}=\cos(2r_\ff).
\end{align}
Hence, we need to rescale the operators 
\begin{align}
    \tilde L_1(\hat\gamma_1):=\frac{L_1(\hat\gamma_1)}{\sqrt{\cos(2r_\ff)}},\quad \tilde L_1(\hat\eta_1):=\frac{L_1(\hat\eta_1)}{\sqrt{\cos(2r_\ff)}},
\end{align}
to obtain properly anti-commuting quadrature operators, defining a bosonic mode. For these, we find
\begin{align}
    \braket{J_\bb|\{\tilde L_1(\hat\gamma_1),\tilde L_1(\hat\eta_1)\}|J_\bb} =\frac1{\cos(2r_\ff)}\,,
\end{align}
showing that the restriction of $J_\bb$ has eigenvalues $\pm\ii\left(\cos(2r_\ff)\right)^{-1}$.

In the limit of $r_\ff\to\pi/4$, the fermionic site mode $(\hat\gamma_1,\hat \eta_1)$ approaches maximal entanglement with the rest of the system, corresponding to an entanglement entropy of one bit, \ie $\log{2}$ in natural units. Consequently, also its bosonic dual system approaches maximal entanglement. However, for the bosonic mode, this means that its entanglement entropy grows without bound, as shown in  Fig.~\ref{fig:sf_sb_dual}. At the point of $r_\ff=\pi/4$, the fermionic mode $1$ would  represent a fermionic subsystem which is maximally entangled with mode $2$. However, such a fermionic mode is not mapped to a valid bosonic subsystem by the identification maps. In fact, for $r_\ff=\pi/4$, the identification map $L_1$ acts as
\begin{align}\label{eq:100}
    L_1(\hat\gamma_1)&=\frac{\cosh(r_\bb)-\sinh(r_\bb)}{\sqrt2} \left(\hat q_1+\hat q_2\right)\,,
    \nn
    L_1(\hat\eta_1)&=\frac{\cosh(r_\bb)-\sinh(r_\bb)}{\sqrt2} \left(\hat p_1-\hat p_2\right)\,,
\end{align}
which are commuting observables, and thus do not define a proper bosonic mode (cf.~\eqref{eq:system_complement}), as also seen by the fact that~\eqref{eq:commutator_on_bosonmode} vanishes.

This example highlights how, in general, fermionic Majorana operators that generate an almost maximally entangled mode are mapped to almost commuting bosonic operators by the identification maps, which in the limit of maximal fermionic entanglement, thus, fail to define a bosonic mode. 
The following Sec.~\ref{sec:Kitaev_2D} showcases a peculiar consequence of this fundamental relationship between highly entangled fermionic modes and their bosonic counterparts in 2D.

\subsection{Application: supersymmetric Kitaev honeycomb model}\label{sec:Kitaev_2D}

\begin{figure*}
\centering
    \subfloat[
    \label{fig:kitaev_diagram_f}]{\includegraphics[width=1.0\linewidth]{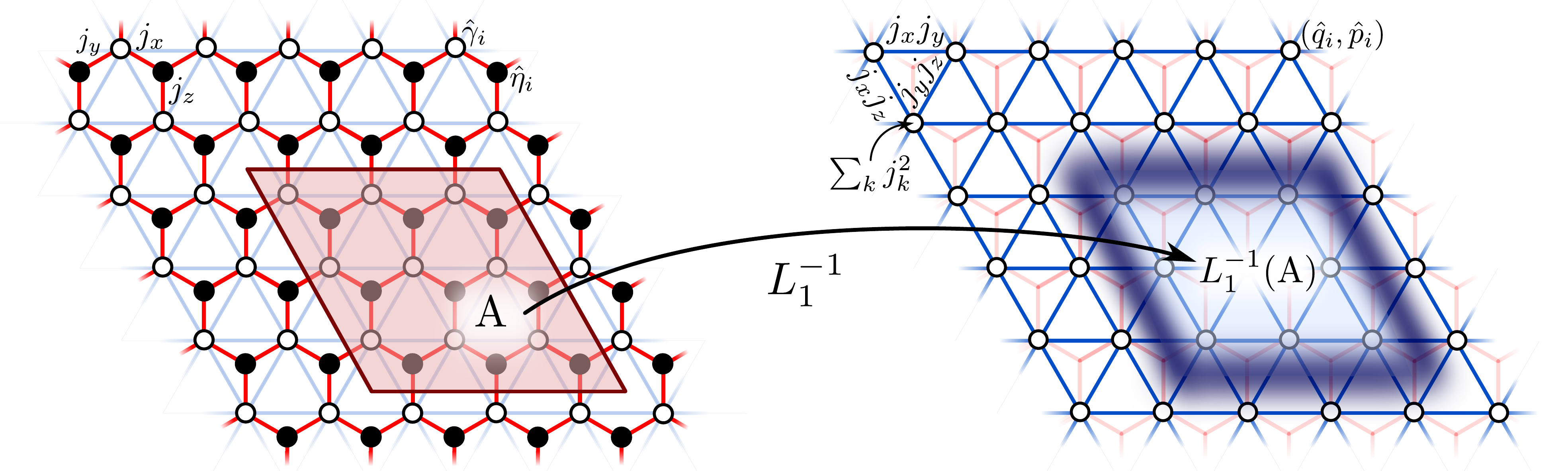}}
    
    \subfloat[
    \label{fig:kitaev2d_entropies}]{\includegraphics[width=.5\linewidth]{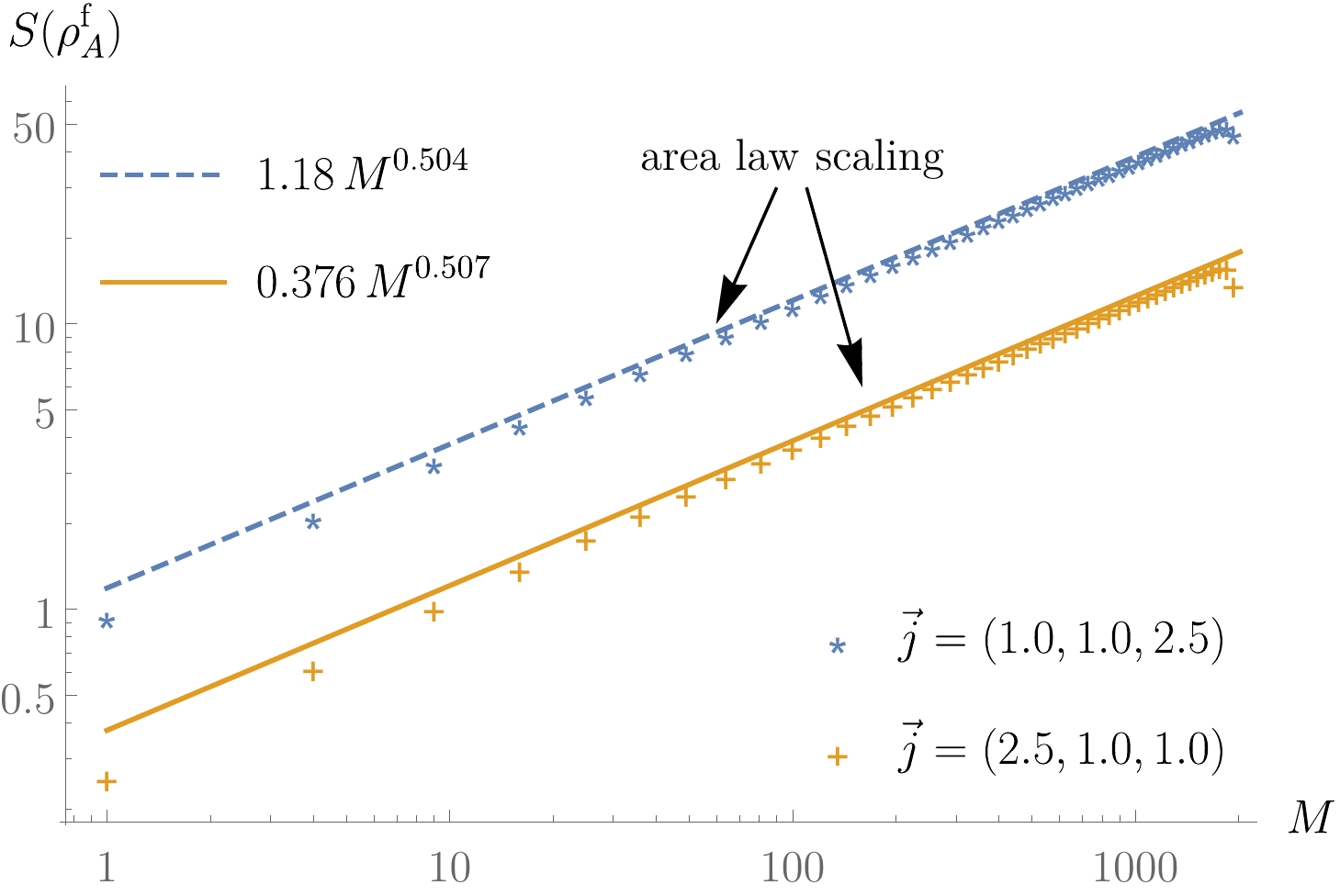}}
    \subfloat[
    \label{fig:kitaev2d_entropies_dual}]{\includegraphics[width=.5\linewidth]{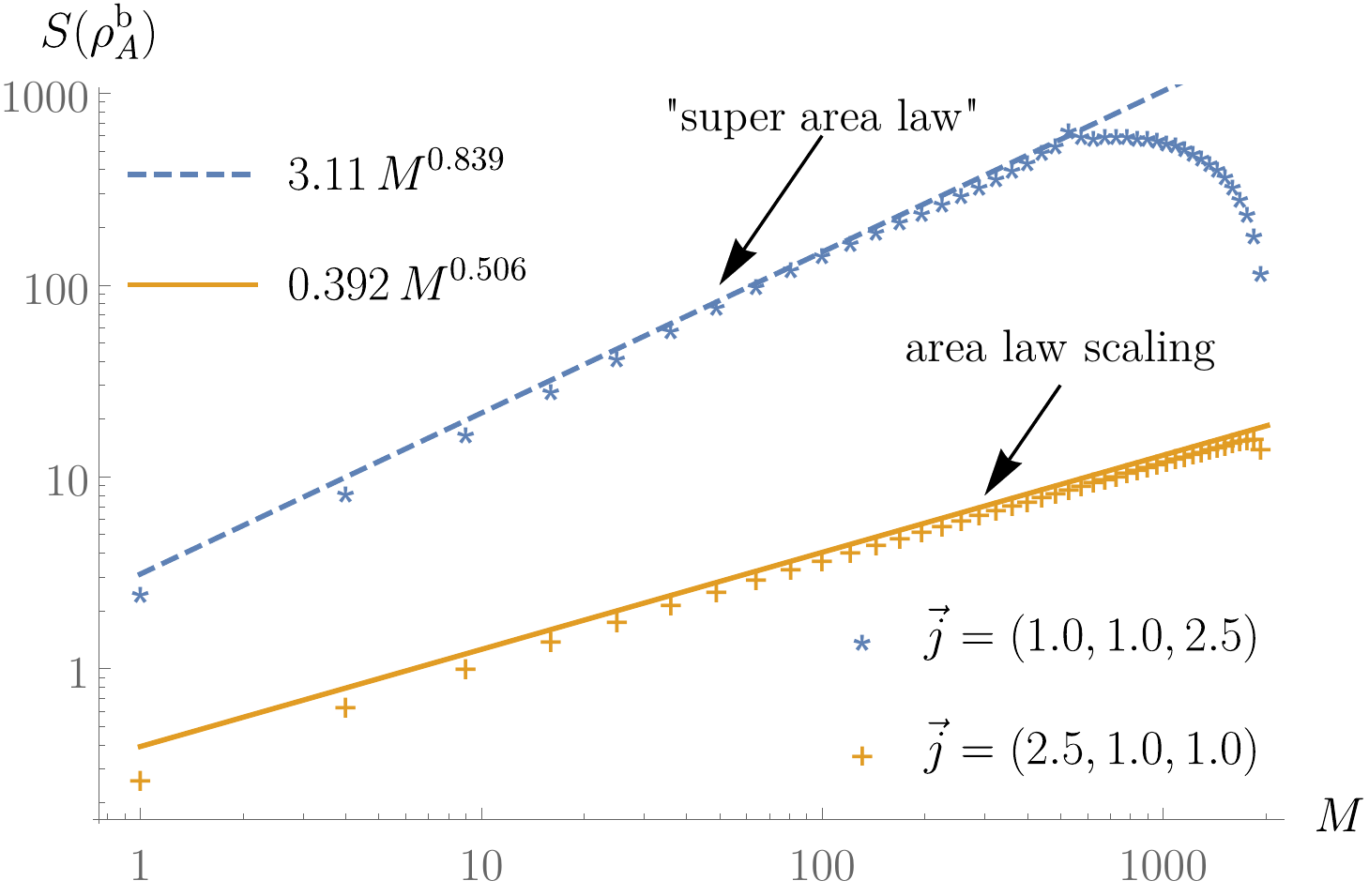}}
    \vspace{-2mm}
    \caption{
    Supersymmetric Kitaev honeycomb model:
    The schematic visualization in (a) represents, on the left, the fermionic honeycomb Kitaev model~\eqref{eq:kitaev2d_Hf} and, on the right,  the bosonic triangular lattice~\eqref{eq:kitaev2d_Hb} which are generated by the supercharge~\eqref{eq:kitaev_2d_Q}.  The map $L_1^{-1}$ identifies fermionic subsystems with bosonic dual subsystems.
    The plots below, in (b) and (c), show numerical calculations of the entanglement entropy of parallelogram-shaped subsystems on the fermionic side, and of their dual bosonic subsystems on the other side, for two different orientations  of the hopping parameters $\vec j=(j_x,j_y,j_z)$.
    For the numerical examples, periodic lattices with a total number of $N=45\times 45=2025$ unit cells were considered. The parallelograms of the subsystems contain $M=m\times m$ modes, \ie $M$ unit cells of the honeycomb lattice, and are oriented such that they do not cut through links with hopping parameter $j_x$.
    The fermionic entropies show good numerical agreement with the area law, which they are known to follow in the thermodynamic limit.
    Depending on the orientation of the couplings $\vec j$ relative to the parallelogram, the dual entropies can follow the area law or a scale much faster.
    }\label{fig:kitaev_2d_fullfigure}
\end{figure*}

In this section, we demonstrate consequences of the derived entanglement duality in the example of the celebrated Kitaev honeycomb model~\cite{kitaev2006anyons}, a spin model with characteristic bond-directional exchanges on the honeycomb lattice (Fig.~\ref{fig:kitaev_diagram_f}), and its supersymmetric extension~\cite{attig2019topological}.
In their gapped phases, both the fermionic and the bosonic lattice of this supersymmetric system exhibit the entanglement area law~\eqref{eq:ent_area_law1}.
Because the identification maps between the fermionic and the bosonic lattice behave local and preserve the shape of subregions of a lattice very well, one may expect also the entropy of these dual subsystems to follow an area law.
However, we show that, in mapping from fermionic subregions to bosonic ones, a peculiar phenomenon can arise where the entanglement entropy of the dual bosonic subsystems scales much faster than its pre-image in the fermionic lattice which follows the area law. This is attributed to the presence of almost maximally entangled modes in the fermionic subsystem.

The analytical solution of the Kitaev honeycomb model is achieved by recasting it in terms of non-interacting Majorana fermions hopping on the same honeycomb lattice (in the
background of a classical (static) $Z_2$ gauge field). The resulting fermionic Hamiltonian reads
\begin{align}\label{eq:kitaev2d_Hf}
 \hat{H}_\ff = -\frac{\ii}{2}\sum_{i,j=1}^N \big( \hat \eta_i {\cal A}^\intercal_{ij} \hat\gamma_j -\hat{\gamma}_i {\cal A}_{ij} \hat\eta_j\big)\,.
\end{align}
Expressed this way, $\hat H_\ff$ describes the hopping of Majorana fermions between the two types of sites of the honeycomb lattice (Fig.~\ref{fig:kitaev_diagram_f}), where each of the Majorana operators $\hat\gamma_i$ and $\hat\eta_i$ resides on  one type of the lattice sites.
The $N\times N$-matrix $\cal{A}$ corresponds to the connectivity matrix of the lattice as depicted in Fig.~\ref{fig:kitaev_diagram_f}, which we consider to be periodic. It involves the hopping strengths along the three bonds around each site of the honeycomb lattice, which we denote by $j_x, j_y, j_z$. The inequality $|j_x|\le |j_y|+|j_z|$ and its cyclic permutations together imply a gapless spectrum of $\hat{H_\ff}$; otherwise, $\hat{H_\ff}$ has a gapped spectrum.
While the phenomena discussed below can arise in both phases, for our numerical results, we will focus on the gapped phase below.

\begin{figure}[t]
\begin{center}
  \includegraphics[width=.98\linewidth]{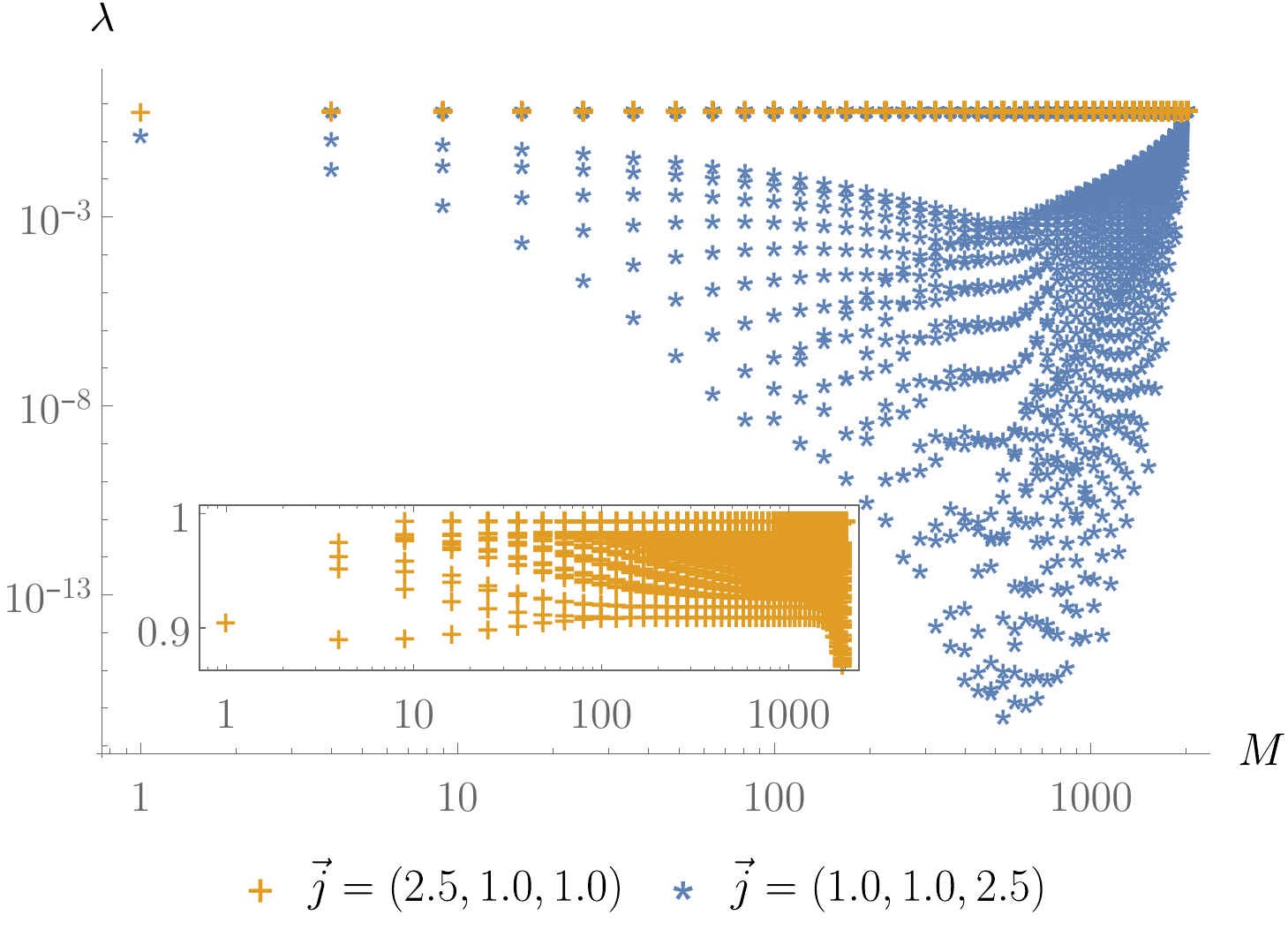}
\end{center}\vspace{-4mm}
\caption{Absolute values of the eigenvalues of the fermionic linear complex structure $J_A^\ff$ on the Kitaev honeycomb lattice restricted to the subsystems of Fig.~\ref{fig:kitaev_2d_fullfigure}. The inset zooms in on the eigenvalues of the first case $\vec j=(2.5,1,1)$, which hardly fall below $|\lambda_i|\approx0.9$. In the second case $\vec j=(1,1,2.5)$, the absolute values decay exponentially with the subsystem size $M$, thus triggering the amplified scaling of the dual entropies in Fig.\ref{fig:kitaev2d_entropies_dual} up until the subsystem exhausts the full lattice of $N=2025$ modes.
}
\label{fig:kitaev_2d_Jspectrum}
\end{figure}

\begin{figure*}
\centering
    \subfloat
    {\includegraphics[width=.45\linewidth,trim={0 0 0.02cm 0},clip]{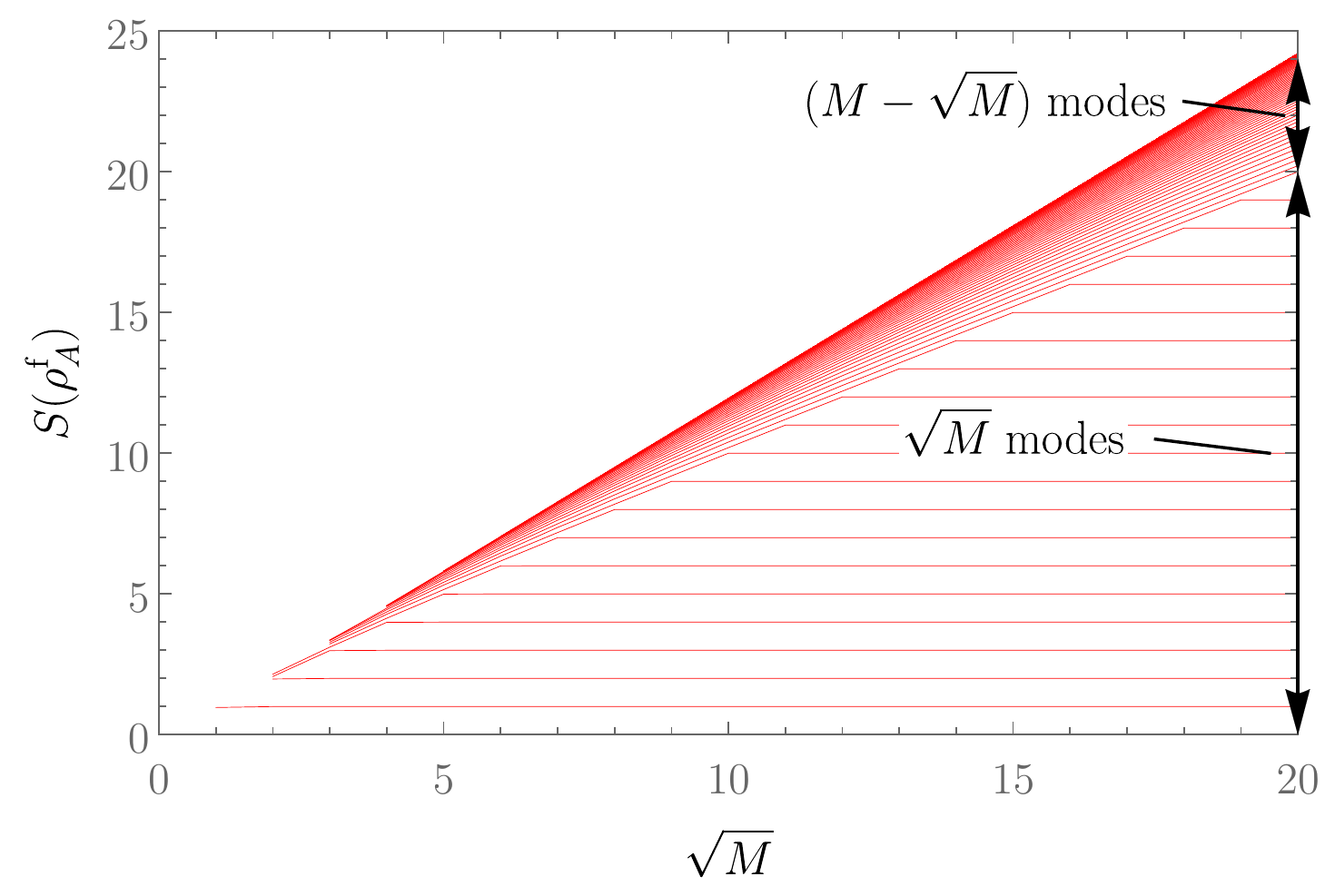}}
    \hspace{.05\linewidth}
    \subfloat
    {\includegraphics[width=.45\linewidth]{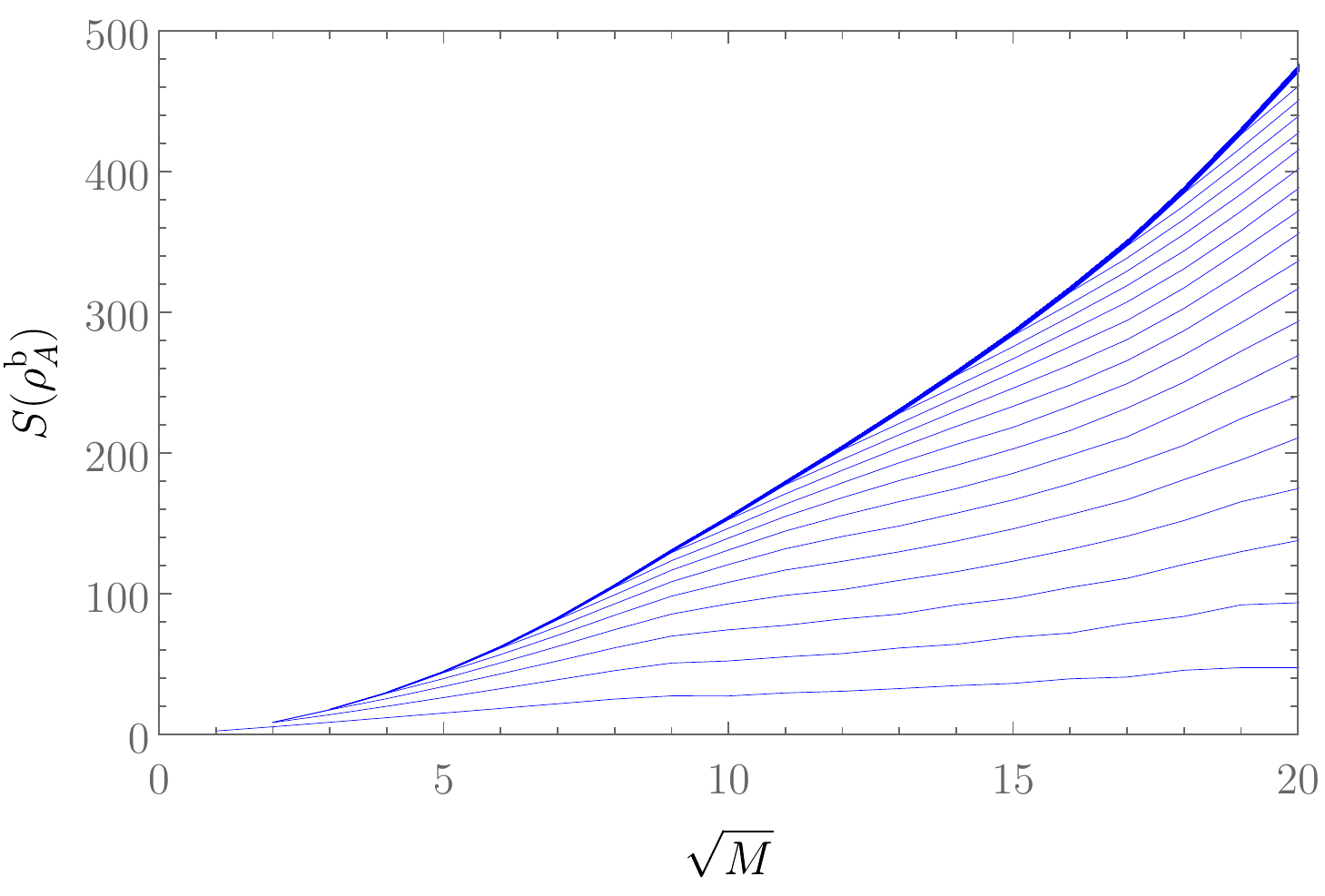}}
    \vspace{-2mm}
    \caption{Stacked plots of the mode-wise contribution to the total entanglement entropy, for the setup of Figs.~\ref{fig:kitaev2d_entropies} and~\ref{fig:kitaev2d_entropies_dual}, for the case $\vec j=(1,1,2.5)$. 
    The left plot refers to the fermionic hexagonal subsystems $A_\ff$ and the right plot to the dual bosonic subsystem $L_1^{-1}(A_\ff)$.
    The difference between the $(i-1)$th and $i$th line gives the contribution of the $i$th normal mode of the subsystem to the total entanglement entropy. Thus, the upper-most line coincides with the respective plot in Figs.~\ref{fig:kitaev2d_entropies} and~\ref{fig:kitaev2d_entropies_dual}. Note that, here, we have changed the horizontal axis to  $\sqrt{M}$, \ie the side length of the parallelogram-shaped fermionic subsystems. The fermionic entropy is dominated by the  $\sqrt M$  normal modes that are almost maximally entangled and the bosonic entropy by their dual modes.
    }\label{fig:kitaev_2d_stacked_plots}
\end{figure*}

A supercharge operator $\hat{\mathcal{Q}}=R_{\alpha a}\hat\xi^\alpha_\ff\hat\xi^a_\bb$ that leads to $\hat{H_\ff}$ being identified with the fermionic part of the supersymmetric Hamiltonian $\hat{H}=\hat{ \mathcal{Q}}^2$ is~\cite{attig2019topological} 
\begin{align}\label{eq:kitaev_2d_Q}
 \hat{\mathcal{Q}} =\sum_{i,j=1}^N \big(\hat \gamma_i {\cal A}_{ij} \hat{q}_j + \hat\eta_i \delta_{ij} \hat{p}_j\big)\,,
\end{align}
which implies a block-diagonal matrix representation of $R_{\alpha a}$. The bosonic part of this Hamiltonian 
\begin{align}\label{eq:kitaev2d_Hb}
    \hat H_\bb=\frac12\sum_{i,j=1}^N  \hat q_i \left(\mathcal{A}^\intercal\mathcal{A}\right)_{ij}\hat q_j +\sum_{i=1}^N\hat p_i^2\,,
\end{align}
corresponds to a triangular lattice of harmonic oscillators, as depicted in Fig.~\ref{fig:kitaev_diagram_f}, or in the appropriate classical limit, to a triangular network of balls and springs~\cite{attig2019topological}.

As previously mentioned, in the gapped phase, both the fermionic and the bosonic lattices of this SUSY Hamiltonian exhibit an area law scaling~\eqref{eq:ent_area_law1} in the entanglement entropy of lattice subregions.
Accordingly, Fig.~\ref{fig:kitaev2d_entropies} shows good agreement of our numerical example with an area law scaling of the entanglement entropy. 
There, we consider a honeycomb lattice which is periodic, with equal side lengths, comprising $N=45\times 45=2025$ unit cells, of two sites each, in total.
From this fermionic lattice, we cut out parallelogram-shaped subsystems of sidelength $m$, \ie containing $M=m\times m$ unit cells, as indicated in Fig.~\ref{fig:kitaev_diagram_f}, and calculate their entanglement entropy $S(\rho_A^\ff)$ with the rest of the lattice.
We compare two different combinations of the hoppings, $\vec j=(j_x,j_y,j_z)=(2.5,1,1)$ and $\vec j=(1,1,2.5)$, with respect to the orientation of the parallelograms. 
These two orientations differ in the type of neighboring sites which the  boundary of the parallelogram separates. 

We emphasize that the absolute shape of the boundary is not crucial here; instead, what types of links (strong or weak) are separated leaves remarkable effects on the scaling of the entanglement entropy in the dual bosonic subsystem, as will be demonstrated shortly. 
In detail, the boundary chosen here only cuts through links with hopping $j_y$ and $j_z$. Thus, in the first case, it only separates sites linked by the two weaker hoppings, whereas in the second case, it separates pairs of sites linked by the strong hopping.
Since the entanglement area law holds independent of the orientation of the couplings, Fig.~\ref{fig:kitaev2d_entropies} shows that the fermionic entanglement does exhibit the scaling as predicted by the area law for both the orientations. The only visible difference is that, in the second case, due to the separation of strongly linked sites, the overall entanglement entropy is about three times larger.

However, one subtle difference between the two orientations which is not evident from the total fermionic entanglement entropy of the parallelogram subsystems is that, in the second case, when strongly linked sites are separated by the boundary of a subsystem, the subsystem develops a significant fraction of ``almost maximally entangled'' fermionic modes that scales as $\sqrt{M}$ for $M$ sites in the subsystem. However, as we shall see, these fermionic modes lead to a vastly different behavior of the entanglement entropy of the dual subsystems in the bosonic lattice associated with the original fermionic parallelograms by the identification maps.

In the periodic lattices considered here, the identification maps behave local in the sense that on-site operators in one lattice are mapped to exponentially localized operators on the supersymmetric partner lattice. Hence, the geometrical appearance of subsystems is well preserved when they are mapped to their dual subsystems in the supersymmetric partner lattice by the identification maps.
At first sight, this may seem to suggest that the entanglement entropy of the dual systems also should exhibit an area law scaling since the entanglement area relation of \eqref{eq:ent_area_law1} holds in both the fermionic and the bosonic lattice we consider.
However, Fig.~\ref{fig:kitaev2d_entropies_dual} demonstrates, for the numerical example introduced before, that the entanglement entropy of fermionic subsystems and their dual bosonic subsystems can scale very differently: depending on the  hoppings $\vec j$ with respect to the parallelogram subsystems, the dual entropy may scale in agreement with an area law scaling or they can scale much faster according to a ``super area law.''

How does this phenomenon arise?
First, let us note that, when the parallelogram in the original fermionic lattice only cuts links with weaker hoppings, the dual entropy follows an area law scaling. The higher scaling of the dual entropy appears when the parallelogram cuts through links with the strongest hopping.
This separation of strongly linked Majorona sites, however, heralds the presence of normal modes in the fermionic parallelogram which are (almost) maximally entangled with the rest of the lattice.
The presence of such modes  is the reason for the observed peculiar scaling of the dual entropies.

In fact, the mathematical explanation for the observed amplified scaling of the dual entropies is rooted in the spectrum of the restricted fermionic linear complex structure $J^\ff_A$. 
Figure~\ref{fig:kitaev_2d_Jspectrum} plots the absolute values of the eigenvalues for the subsystems considered in Fig.~\ref{fig:kitaev_diagram_f} for the two distinct orientations of the hopping mentioned before.
In the first case, where the parallelograms do not cut through any strong links, the eigenvalues roughly lie in the interval $0.9 <|\lambda_i|\leq 1$, as seen in the inset of Fig.~\ref{fig:kitaev_2d_Jspectrum}.
As is evident from Fig.~\ref{fig:sf_sb_dual}, in this regime, the entanglement entropy of each of the eigenmodes of $J^\ff_A$, \ie the normal modes of the subsystem, is almost the same as the entanglement entropy of their dual bosonic modes. Thus, in the first case, the entanglement entropies for the fermionic subsystems and their bosonic duals are almost the same and follow the same scaling.

In contrast, in the second case, the spectrum of $J_A^\ff$ exhibits a certain number of eigenvalues which are very small or almost zero. Note that the number of these pairs of eigenvalues $\pm\ii\lambda_i$ that fall below $\lambda_i\lesssim 0.1$ corresponds exactly to the side length of the parallelograms, \ie is half of the number of strong links which the parallelogram  cut through.
The normal modes corresponding to these eigenvalues thus share almost maximal entanglement with the rest of the system, \ie the complement of the region surrounded by the parallelogram.
In terms of their entanglement entropy, as discussed before and also evident from Fig.~\ref{fig:sf_sb_dual}, the fermionic normal modes approach the maximum value of one bit entanglement entropy as $\lambda_i\to0$  following~\eqref{eq:entropy_functions}, whereas the entanglement entropy of their dual bosonic modes diverges as $\lambda_i^{-1}\to\infty$.

As a result, the total entanglement entropy of the dual bosonic subsystem scales much faster with its subsystem size than the original fermionic system does.
This effect is visualized in Fig.~\ref{fig:kitaev_2d_stacked_plots}, whose stacked plots show the mode-wise contribution of the normal modes to the total entropy of the fermionic subregions and their duals in the bosonic lattice.
On the fermionic side, the individual contributions are bounded by one bit per mode; thus, their summed contribution still results in a growth linear in the perimeter of the parallelogram. However, on the bosonic side, the individual contribution from each normal mode continues to grow as the system size increases, resulting in a higher scaling of the entanglement entropy than that predicted by the area law.
Let us emphasize that the total number of low-lying fermionic eigenvalues scaling as $\sqrt{M}$ (with $M$ being the subsystem size) alone is \emph{not} sufficient to give rise to a ``super area law'' on the bosonic side but also that these low-lying values actually decay towards zero. If they were bounded by some $\lambda_{\min}$, such that $\lambda_i\geq \lambda_{\text{min}}>0$, the entropy of each dual bosonic mode would be upper bounded by $s_\bb(1/\lambda_{\min})$,  resulting again in the conventional $\sqrt{M}$ scaling of the area law.

The peculiar phenomenon observed above can be viewed as a direct physical instance of the minimal two-mode example in Sec.~\ref{sec:two_mode_duality} taking place at the edge of the subsystem:
Every time its boundary cuts through a pair of strong links (on opposite sites of the parallelogram cutout), the subsystem exhibits a strongly entangled normal mode (corresponding to an almost vanishing eigenvalue of the restricted complex structure).
These normal modes are highly localized at the edge of the subsystem and share no entanglement with any mode inside the subsystem but with those lying on the complement of the subsystem. In fact, the normal modes of the subsystem, which carry entanglement, form pairs with those from the complement such that each normal mode is entangled with exactly one partner (normal) mode of the complement.\footnote{These pairs are connected by the complex structure $J_\ff$ of the ground state.}
The partner normal modes of the highly entangled subsystem normal modes are localized right outside the subsystem. Thus, a pair of partner modes (one inside the subsystem and one outside) forms a two-mode subsystem, localized in the immediate neighborhood of the  boundary of the subsystem, which is not entangled with the rest of the system but in a pure two-mode squeezed state on its own.
The identification maps now map each pair of such fermionic normal mode partners to a pair of bosonic normal mode partners, one inside the dual subsystem and one outside. Due to the locality properties of the identification maps, their joint support on the bosonic lattice sites is closely related to the shape of the fermionic pair.

In this mapping, pair by pair, the same mechanism as discussed in \eqref{eq:100} takes place.
The Majorana operators of the fermionic subsystem normal mode are mapped to a pair of bosonic observables which are almost commuting, thus defining a highly entangled bosonic mode.
Such almost commuting bosonic observables need not be spatially separated on the lattice, but they can have equal support on the same lattice sites, as \eqref{eq:100} demonstrates: There, both bosonic observables have equal support on both of the two modes; however, one quadrature is proportional to $\hat q_1+\hat q_2$ but the other to $\hat p_1-\hat p_2$; thus, they commute.

Because of such localized and highly entangled bosonic modes, it is possible for the dual bosonic subsystems, despite being well localized, to exhibit a scaling of entanglement entropy that exceeds the area law of the original fermionic lattice. The entanglement area law assumes the subsystem division being a direct sum of individual lattice sites, \ie in a bosonic system, the quadarature operators $\hat q_i$ and $\hat p_i$ either both belong to the subsystem or they both do not.
In contrast, the boundary between the dual bosonic subsystems and the rest of the (bosonic) lattice considered in this example may well separate different linear combinations of the onsite bosonic operators. 

\section{Discussion}\label{sec:discussion}
In this article, we study the entanglement properties of bosonic and fermionic Gaussian states that are related via supersymmetry, in other words, belong to Hamiltonians which are supersymmetric partners of each other. After reviewing a unified framework to describe these states in terms of Kähler structures, we prove the main result of this article in Proposition\ \ref{prop:entanglement-duality}, which relates the bosonic and the fermionic entanglement spectrum of a chosen subsystem in a supersymmetric Gaussian state. The result is based on supersymmetric identification maps that are constructed from the supercharge operator $\hat{\mathcal{Q}}$. They enable us to uniquely identify subsystems both bosonic and  fermionic, which we refer to as dual to each other. In Proposition~\ref{prop:thermal-duality}, we extend the said duality to include thermal states associated with supersymmetric Hamiltonians, for which we find the same relationship between the bosonic and the fermionic thermal states as for the reduced states in the subsystems.

The rest of the article illustrates this result and its implications in supersymmetric lattice models. 
In particular, we investigate to what extent identification maps  constructed from a \emph{local} supercharge operator preserve this locality, \ie to what extent a local subsystem on the bosonic side is identified with a local subsystem on the fermionic side and vice versa. 
This is important to explain why our abstract duality is of relevance when studying physical systems: It shows a simple relation between the entanglement of bosonic and fermionic subsystems that can both be thought local in a precise way.

The examples in this work suggest that, for SUSY lattice Hamiltonians, the locality properties of the identification maps are related to the boundary conditions of the lattice and the presence of edge modes. 
For example, in the one-dimensional open chain of Sec.~\ref{sec:Kitaev_1D_example}, the identification maps featured highly non-local behavior in the topological phase. However, in the same system with periodic boundary conditions (obtained by extending the supercharge \eqref{eq:Kitaev_1D_supercharge} to be translation-invariant), the identification maps behave rather local, even deep in the topological phase.

In the context of localized subsystems, a peculiar consequence of the entanglement duality is the appearance of ``super area law'' behavior in the entanglement entropy of bosonic subsystems dual to subsystems with certain shapes on the fermionic side, seen in Sec.~\ref{sec:Kitaev_2D}. This phenomenon is related to the appearance of almost maximally entangled modes in the fermionic subsystem, for which the spectrum of the fermionic linear complex structure nearly vanishes. The entanglement duality then implies an unbounded growth of entanglement for the dual bosonic system. 
Since it is well known~\cite{eisert2010colloquium} that the entanglement entropy associated with a ground state of a gapped and local Hamiltonian (bosonic or fermionic) satisfies an area law, this raises the question of how occasions where our entanglement duality relates an area law on the fermionic side with a ``super area law'' on the bosonic side for the respective ground states of a local supersymmetric Hamiltonian (such as the honeycomb model considered in~\ref{sec:Kitaev_2D}) can appear. The answer to this question lies in the type of bosonic subsystem that arises under the duality for a fermionic subsystem with large entanglement. 
We saw in Sec.~\ref{sec:two_mode_duality} how there can be an arbitrary amount of entanglement associated with a single bosonic mode due to choosing a subsystem that effectively partially separates a quadrature operator $\hat{q}_i$ from its canonically conjugate operator $\hat{p}_i$. Such types of subsystems are typically not considered in the context of studying area laws, and it is not surprising that the standard results on area laws for the ground states of gapped local Hamiltonians do not apply to them.

At this stage, it is a natural question to what extent such dual subsystems, identified by the supersymmetric duality as in Sec.~\ref{sec:Kitaev_2D}, should be considered a physical reality  or a mathematical concept. 
The answer will highly depend on the concrete physical realization of the model. For example, if the bosonic degrees of freedom were represented by photons,  such a subsystem could at least in principle be probed by applying an appropriate sequence of Gaussian transformations (implemented by standard linear optics devices), as long as $\hat{q}_i$ and $\hat{p}_i$ are not fully separated, \ie $[\hat{q}_i,\hat{p}_i]\neq 0$. 
On the other hand, since the implementation of large squeezing, as discussed in Sec.~\ref{sec:two_mode_duality}, is known to be very challenging, accessing the subsystems may well be physically infeasible.

In keeping with the study of topological properties of translation-invariant SUSY lattice Hamiltonians in arbitrary dimensions, a generalization of the identification maps to higher dimensional lattices certainly constitutes a promising avenue to explore. Here, we would like to highlight \cite{attig2017classical}, where a spin-fermion correspondence, very much in the same spirit of our SUSY map,  has been worked out engaging three-dimensional lattice models as well. Entanglement properties of a three-dimensional generalization of the Kitaev honeycomb model have also been studied \cite{mondragon2014entanglement}. Other variants of three-dimensional Kitaev spin liquids exist \cite{o2016classification, eschmann2020thermodynamic} with entanglement properties hitherto unexplored; our dualities find a straightforward application therein.

\begin{acknowledgements} 
RHJ gratefully acknowledges support by the Wenner-Gren Foundations. 
LH gratefully acknowledges financial support by the Alexander von Humboldt Foundation. 
KR thanks the sponsorship, in part, by the Swedish Research Council.
\end{acknowledgements}

\bibliographystyle{unsrt}
\bibliography{references}

\end{document}